\documentclass[11pt]{article}
\usepackage[margin=1in]{geometry}
\usepackage{mathtools,amsfonts,amsthm,amssymb,bbm}
\usepackage{enumitem}
\usepackage{xcolor}
\usepackage[backref,colorlinks,citecolor=blue,linkcolor=magenta,bookmarks=true]{hyperref}
\usepackage[nameinlink]{cleveref}
\usepackage[square,sort,numbers]{natbib}
\usepackage[algo2e,ruled,linesnumbered]{algorithm2e}

\title{Leakage-Robust Bayesian Persuasion\footnote{A preliminary version was accepted at the 26th ACM Conference on Economics and Computation (EC 2025).}}
\author{
Nika Haghtalab\thanks{UC Berkeley. Email: \texttt{nika@berkeley.edu}.}
\and
Mingda Qiao\thanks{MIT. Email: \texttt{mingda.qiao.cs@gmail.com}. Part of this work was done while the author was at UC Berkeley.}
\and
Kunhe Yang\thanks{UC Berkeley. Email: \texttt{kunheyang@berkeley.edu}.}
}
\date{}
\usepackage[nameinlink]{cleveref}
\usepackage{graphicx}
\usepackage{multirow,hhline,tabularx,xcolor,wrapfig,tikz,cancel,nicefrac}
\usepackage{thm-restate}
\usetikzlibrary{shapes.geometric}
\usepackage{wrapfig}
\usepackage{bbm}
\usepackage{thmtools}
\usepackage{thm-restate}
\usepackage{enumitem,array,booktabs}
\newcommand{\1}[1]{\mathbbm{1}\left[#1\right]}
\newcommand{\A}{\mathcal{A}}

\newcommand{\support}{\mathrm{supp}}

\newcommand{\cG}{\mathcal{G}}

\newcommand{\D}{\mathcal{D}}

\newcommand{\signalspace}{\mathcal{S}}

\newcommand{\eps}{\epsilon}
\newcommand{\Ex}[2]{\operatorname*{\mathbb{E}}_{#1}\left[#2\right]}

\newcommand{\OPT}{\mathsf{OPT}}

\newcommand{\OPTprivate}{\OPT^{\mathsf{private}}}

\newcommand{\OPTpublic}{\OPT^{\mathsf{public}}}
\newcommand{\OPTpersuasive}{\OPT^{\mathsf{persuasive}}}
\newcommand{\maxminone}{\mathsf{Maxmin}^{1}}
\newcommand{\maxmintwo}{\mathsf{Maxmin}^{2}}
\newcommand{\OPTexpected}{\OPT^{\mathsf{expected}}}
\newcommand{\kbroadcast}{k\text{-}\mathsf{broadcast}}
\newcommand{\kstar}{k\text{-}\mathsf{star}}
\newcommand{\kclique}{k\text{-}\mathsf{clique}}
\newcommand{\kErdosRenyi}{k\text{-}\mathsf{Erd\ddot{o}s}\text{-}\mathsf{R\acute{e}nyi}}

\newcommand{\pr}[2]{\Pr_{#1}\left[#2\right]}
\newcommand{\R}{\mathbb{R}}

\newcommand{\subjectto}{\text{ Subject to }}

\newcommand{\unif}{\textsf{Unif}}

\newcommand{\prior}{\tau}

\newcommand{\bmuzero}{\boldsymbol{\mu}_{\boldsymbol{0}}}
\newcommand{\bmuone}{\boldsymbol{\mu}_{\boldsymbol{1}}}
\newcommand{\bmu}{\boldsymbol{\mu}}
\newcommand{\bmuomega}{\boldsymbol{\mu}_{\boldsymbol{\omega}}}
\newcommand*\circled[1]{\tikz[baseline=(char.base)]{
            \node[shape=circle,draw,inner sep=0.8pt] (char) {#1};}}

\newcommand{\PoWR}{\mathsf{PoRP}}
\newcommand{\PoDR}{\mathsf{PoRU}}

\theoremstyle{plain}
\newtheorem{theorem}{Theorem}[section]
\newtheorem{lemma}[theorem]{Lemma}

\newtheorem{proposition}[theorem]{Proposition}

\newtheorem{definition}{Definition}[section]
\newtheorem{remark}[theorem]{Remark}
\newtheorem{example}[theorem]{Example}

\crefname{definition}{definition}{definitions}
\Crefname{definition}{Definition}{Definitions}
\crefname{prop}{proposition}{propositions}
\Crefname{Prop}{Proposition}{Propositions}
\Crefname{cor}{Corollary}{Corollaries}
\crefname{lemma}{Lemma}{Lemmas}
\crefname{section}{Section}{Sections}
\crefname{subsubsubsection}{Section}{Sections}
\crefname{remark}{Remark}{Remarks}
\crefname{figure}{Fig.}{Figs.}
\crefname{table}{Table}{Tables}
\Crefname{lemma}{Lemma}{Lemmas}
\crefname{theorem}{Theorem}{Theorems}
\Crefname{theorem}{Theorem}{Theorems}
\crefname{algo}{Algorithm}{Algorithms}

\begin{document}
\thispagestyle{empty}

\maketitle

\begin{abstract}
    This paper introduces the concept of
	leakage-robust Bayesian persuasion.
	Situated between \emph{public} Bayesian persuasion~\cite{KG2011bayesian} (and its multi-receiver variants~\cite{public-almost-optimal,xu2020tractability}) and \emph{private} Bayesian persuasion~\cite{AB19},
	leakage-robust persuasion considers a setting where one or more signals privately communicated by a sender to the receivers may be leaked.
	We study the design of leakage-robust Bayesian persuasion schemes and quantify the price of robustness using two formalisms:

	\begin{itemize}
        \item[-]     The first notion, \emph{$k$-worst-case persuasiveness}, requires a signaling scheme to remain persuasive as long as each receiver observes no more than $k$ leaked signals from other receivers. We quantify the Price of Robust Persuasiveness ($\PoWR_k$)--- i.e., the gap in sender's utility as compared to the optimal private persuasion scheme---as  $\Theta(\min\{2^k,n\})$ for supermodular sender utilities and $\Theta(k)$ for submodular or XOS sender utilities, where $n$ is the number of receivers. This result also establishes that in some instances, $\Theta(\log k)$ leakages are sufficient for the utility of the optimal leakage-robust persuasion to degenerate to that of public persuasion.
 		\item[-]  The second notion, \emph{expected downstream utility robustness}, relaxes the persuasiveness requirement and instead considers the impact on sender's utility resulting from receivers best responding to their observations. By quantifying  the Price of Robust Downstream Utility ($\PoDR$) as the gap between the sender's expected utility over the randomness in the leakage pattern as compared to private persuasion, our results show that, over several natural and structured distributions of leakage patterns, $\PoDR$ improves $\PoWR$ to $\Theta(k)$ or even $\Theta(1)$, where $k$ is the maximum number of leaked signals observable to each receiver across leakage patterns in the distribution.
    \end{itemize}

	En route to these results, we show that \emph{subsampling} and \emph{masking} serve as general-purpose algorithmic paradigms for transforming any private persuasion signaling scheme to one that is leakage-robust, with minmax optimal loss in sender's utility.
\end{abstract}

\thispagestyle{empty}
\clearpage
\pagenumbering{arabic}

\section{Introduction}

Bayesian persuasion, introduced by \citet{KG2011bayesian}, is a framework for studying the fundamental problem of  persuading rational agents by
controlling their informational environment.
For example, consider a prosecutor who knows more about whether a defendant is guilty and aims to convince a jury to  convict. By carefully designing the investigation---deciding whom to subpoena, what questions to ask an expert witness, and which forensic tests to conduct---the prosecutor can influence the jury's belief in favor of conviction.\footnote{A single-receiver version of this example, in which a prosecutor aims to persuade a judge, serves as the motivating example in \cite{KG2011bayesian}.
}
Similarly, in online advertising, a seller with more information about a product's quality seeks to persuade potential buyers to make a purchase. By designing advertisements that selectively highlight certain aspects of the product, the seller can influence buyers’ purchasing decisions.
Bayesian persuasion abstracts these choices by modeling the actions of the prosecutor or seller as the design of \emph{signaling schemes}, i.e., structured distributions of information or action recommendations conditioned on the true state of the world, that only partially reveal the truth. \emph{Bayesian} receivers, such as buyers and the jury, then update their beliefs according to the signals generated from these schemes. The design and effectiveness of these signaling schemes has been the main subject of research on Bayesian persuasion.

The effectiveness of Bayesian persuasion largely depends on whether signals can be communicated publicly or privately. In \emph{public Bayesian persuasion}~\cite{public-almost-optimal,xu2020tractability}, all signals are publicly observable to all receivers, whereas in \emph{private Bayesian persuasion}~\cite{AB19}, the sender can communicate personalized signals to each receiver through private communication channels. It is not hard to see that a
prosecutor who could communicate privately with each jury member\footnote{An act that is strictly prohibited in the US court system!} can influence the outcome  more effectively
by tailoring the investigation details for each member.
Similarly, in online advertising, the seller can persuade more customers to make a purchase by tailoring advertisements according to each customer's preferences. Indeed, in some cases, the gap in the effectiveness of private versus public persuasion can be as high as the Price of Anarchy of the underlying game~\cite{nachbar2022power}. This superior performance of private Bayesian persuasion makes it particularly appealing to study.

In practice, however, the communication channels used for private persuasion are often not fully private. Signal leakages---that one receiver shares their private signal with others---are common and pose significant threats to the effectiveness of private persuasion. For example, in the online advertising scenario discussed above, although most customers may keep their personalized advertisements private, some might share them on social media, effectively creating a signal leakage. This allows other receivers who may have received different signals to access additional information and update their posterior beliefs, which in turn changes their actions. In particular, certain leaked signals---such as a conservative advertisement intended for a more skeptical customer---can drastically change the beliefs and actions of other receivers, even if the leakage involves only a single signal. Such signal leakages violate the fundamental assumptions of private persuasion and pose challenges to its effectiveness.

In this paper, we study the robustness of private Bayesian persuasion to signal leakages.
We explore this question using two formalisms.
In the first formalism, we look for signaling schemes that must remain persuasive under leakage of even a \emph{worst-case} choice of signals.
In the second formalism,
we relax the requirement of worst-case persuasiveness and instead focus on the \emph{average-case} downstream effect of leakages in the sender's utility.
In both formalisms, we are interested in algorithmic principles for designing sender-optimal signaling schemes and studying the gaps in sender's utility.
When requiring worst-case persuasiveness, we ask how much utility the sender must sacrifice in order to keep the signaling scheme persuasive and how many signal leakages can we tolerate before the utility of robust private persuasion collapses to that of public persuasion.
When studying the expected downstream effect of leakages, we look for designing schemes that achieve a high expected utility under specific distributions over leakage patterns and quantify the gaps therein.

\subsection{Demonstrative Examples}
\label{sec:example}

In this self-contained section, we give examples of how leakages threaten the effectiveness of private persuasion and demonstrate a path forward for overcoming these challenges.

\paragraph{Optimal private persuasion.}
For this section, we use the canonical example of private persuasion in online advertising.
Consider a seller (the sender) who aims to persuade three potential buyers (the receivers) to purchase a product by sending them personalized advertisements. The quality of the product can be either \emph{good} ($\omega_{G}$) or \emph{bad} ($\omega_{B}$), and each buyer $i\in\{1,2,3\}$ decides to purchase based on their perceived quality of the product. Specifically, buyer $i$ will purchase the product if and only if
$\Pr[\omega_B]\le p_i$,
where $p_i$ is their personalized threshold representing their willingness to take a risk.\footnote{While buyers’ behavior is often described using utility functions, this simpler threshold-based characterization is an equivalent characterization in the advertising setting. See \Cref{footnote:threshold} for details.}
On the other hand, the sender's utility depends on the subset of buyers that purchase the product (denoted by $S$). In particular, take sender's utility to be $f(S)=|S|-c(|S|)$, where each item brings in one unit of revenue and $c(|S|)$ is the cost of producing $|S|$ items. The cost function $c(\cdot)$ is typically concave, reflecting the common observation that the marginal cost per additional item decreases as production volume increases.
For this example, we consider the specific cost function $c(|S|)=0.5\cdot|S|^{0.5}$.

To maximize their utility, the sender designs a signaling scheme---a randomized mapping from the product's state to the advertisements $(s_1,s_2,s_3)$ sent privately to each receiver.\footnote{Note that the nonlinearity of the sender's utility function requires optimizing the joint distribution of signals across all receivers, rather than treating each receiver independently.}
At a high level, signaling schemes are designed to influence the receiver's posterior beliefs about the product's quality, thereby shaping their purchase decisions in a way that maximizes the sender's utility.
By the well-known revelation principle, we restrict attention to persuasive direct signaling schemes where each signal $s_i\in\{0,1\}$ serves as a recommendation for buyer $i$ to adopt the product. These recommendations are persuasive if $s_i=1$ leads the buyer $i$ to believe that the bad state is unlikely ($\Pr[w_B\mid s_i=1]\le p_i$), resulting in a purchase, while $s_i=0$ suggests otherwise ($\Pr[w_B\mid s_i=0]> p_i$).
Using Bayes' rule, the persuasive condition for $s_i=1$ can be equivalently written as
\begin{align}
    \Pr\left[s_i=1\mid \omega_B\right]\le \theta_i\cdot\Pr\left[s_i=1\mid \omega_G\right],
    \label{example-cond-persuasive}
\end{align}
where $\theta_i$ is a parameter that only depends on each receiver's intrinsic threshold $p_i$ and the prior distribution (see \Cref{eq:persuasion-level-def} for its formal definition). This parameter is commonly referred to as \emph{persuasion level} in the literature~\cite{AB19}.

In this example, we assume the three receivers have persuasion levels $\theta_1=1,\theta_2=0.5$, and $\theta_3=0.001$, reflecting their different tolerance for bad product quality. Receiver~$1$ is highly risk-tolerant and will purchase even if the sender always recommends purchase in the bad state. Receiver~$2$ is moderately cautious and requires the sender's probability of recommending purchase in the bad state to be no more than half of that in the good state. Receiver~$3$, in contrast, is highly skeptical and will purchase only if the sender rarely recommends purchase in the bad state.

In the fully private setting, \citet{AB19} show that the optimal signaling scheme $\bmu^\star$ should take the following form: in the good state, it recommends all receivers to purchase by sending $s_i=1$ to all $i$.
In the bad state, the optimal scheme sends positive signals to $\{1\}$ with probability $\theta_1-\theta_2=0.5$, to $\{1,2\}$ with probability $\theta_2-\theta_3=0.499$, and to $\{1,2,3\}$ with probability $\theta_3=0.001$. Mathematically, the signaling scheme can be written as
\begin{align*}
    \begin{cases}
        \bmu_{G}^\star(\{1,2,3\})=1;\\
        \bmu_{B}^\star(\{1\})=0.5;\quad
        \bmu_B^\star(\{1,2\})=0.499;\quad
        \bmu_B^\star(\{1,2,3\})=0.001,
    \end{cases}
    \tag{Optimal private scheme}
\end{align*}
where $\bmu_{G}^\star$ and $\bmu_{B}^\star$
denote the scheme as conditioned on the true state being \emph{good} or \emph{bad}, respectively.
One can easily verify that the above scheme is persuasive as it satisfies \Cref{example-cond-persuasive}.

\paragraph{Lack of persuasiveness and loss of sender's utility.}
However, this optimal private scheme is no longer persuasive even when a single signal leaks. To see why, consider the case when the product is bad: there is a $0.999$ probability that receiver $3$ (the most skeptical, with $\theta_3=0.001$) receives a signal of $s_3=0$. If receiver $3$ leaks this signal publicly, receivers $1$ and $2$ will immediately infer that
the product is bad and ignore any positive signal recommending purchase.
This is because, in the good state, the optimal scheme would always send $s_3=1$.
Mathematically, for $i\in\{1,2\}$, the leakage makes the following equation invalid:
\begin{align}
    \Pr\left[s_i=1, s_3=0\mid \omega_B\right]\le \theta_i\cdot\Pr\left[s_i=1, s_3=0\mid \omega_G\right].
    \label{example-condition-leakage}
\end{align}
As a result, the optimal scheme becomes unpersuasive.

As a consequence of the lack of persuasiveness, the sender loses a significant fraction of utility. For simplicity, we focus on the sender's utility under the bad state $\omega_B$.\footnote{It is easy to match the utility in state $\omega_G$ using a full revelation scheme. See, for example, \Cref{thm:gap-persuasive-general} for more details.} Without any leakage, the optimal scheme $\bmu^\star$ achieves an expected utility of $V^\star_B=0.5\cdot f(\{1\})+0.499\cdot f(\{1,2\})+0.001\cdot f(\{1,2,3\})\approx0.9$. However, if receiver $3$ leaks their signal, the sender's utility drops drastically to only $0.001\cdot f(\{1,2,3\})\approx 0.002$. To address this issue, we will introduce a natural modification to $\bmu^\star$ that remains persuasive in the presence of leakages while preserving a good fraction of the sender's utility.

\paragraph{Restoring sender's utility via  leakage-robust persuasive schemes.}
The primary reason that $\bmu^\star$ lacks robustness is that zero signals are never sent under the good state $\omega_G$. This means that even a single leaked zero signal (e.g., $s_3=0$) violates \Cref{example-condition-leakage} as it causes the right-hand side to equal $0$. To address this issue, a natural modification is to introduce randomness into the signaling scheme under state $\omega_G$ by sending occasional zero signals to a randomly \emph{subsampled} subset of receivers.
Consider the subsampling rate $0.5$, where each receiver $i\in\{1,2,3\}$ independently receives a random signal with value $0$ or $1$, each with equal probability $0.5$.
Under this scheme, the right-hand side of \Cref{example-condition-leakage} remains positive and equals $0.5^2=0.25$ for all possible combinations of two signal values.
Therefore, to satisfy \Cref{example-condition-leakage}, it suffices to modify $\bmu_B^\star$ and scale down the probability of sending all nonempty subsets of positive signals by the same factor $0.25$.
We define the modified signaling scheme $\bmu$ as:
\begin{align*}
    \begin{cases}
        \bmu_G(S)=0.5^3,\quad \forall S\in 2^{\{1,2,3\}};\\
        \bmu_B(S)=0.5^2\cdot \bmu_B^\star(S),\quad \forall S\neq\emptyset.
    \end{cases}
    \tag{Persuasive under any $1$ leakage}
\end{align*}
Since the modified scheme $\bmu$ satisfies \Cref{example-condition-leakage}, it guarantees receivers $1$ and $2$ to continue adopting even after observing a leaked signal of $s_3=0$.
It can be checked that the same persuasiveness holds also for the cases where receiver $1$ or $2$ leaks signals. As a result, the modified scheme remains persuasive up to one leaked signal.

We now analyze the sender's utility under $\bmu$. As before, we focus on the utility under state $\omega_B$, since the utility under $\omega_G$ can be easily matched via full revelation.
Using the modified scheme $\bmu$ which satisfies $\bmu_B(S)=0.5^2\cdot \bmu_B^\star(S)$, this scheme secures an expected utility of $0.5^2 \cdot V^\star_B\approx 0.22$, a significant improvement over  $0.002$ achieved with the unmodified scheme.

In this paper, we extend the \emph{subsampling} approach to scenarios where more than one receiver leaks their signals. When up to $k$ signals are leaked, our technique provides a general reduction to robust signaling schemes that remain persuasive and guarantee a fraction of $0.5^{k+1}$ of the sender's optimal utility. We also prove that this factor is tight and cannot be improved for certain classes of supermodular sender utilities. On the other hand, when the sender's utility is more structured (such as submodular), we show that a more biased subsampling rate of $1/k$ can be used to achieve a more favorable scaling factor of $\Theta(1/k)$, significantly improving the sender's utility from $\Theta(0.5^{k})$ to $\Theta(1/k)$ fraction of the optimal utility without leakages.

\paragraph{Unpersuasive leakage-robust schemes with improved expected sender's utility.}
While the above analysis focuses on the persuasiveness against worst-case leakages, there are many cases when the leakage patterns are stochastically drawn from structured families of leakages, among which not all leakages are harmful. For example, consider the random leakage model where a randomly chosen receiver broadcasts their signal to the other receivers. It turns out that when receiver $1$ broadcasts their signal, it does not change the posterior distribution of the other two receivers, as receiver $1$ always receives a positive signal which contains no additional information about the product state. This allows the sender to secure at least $\frac{1}{3}$ fraction of the optimal utility. In fact, we can boost this fraction to reach approximately $\frac{2}{3}$ using the idea of \emph{masking}: we mask out the information carried by receiver $3$'s signal by making the signal deterministically negative. Then, the sender only loses a small utility of $0.001\cdot f(\{1,2,3\})$, but secures another $\frac{1}{3}$ fraction of the sender's utility when receiver $3$ leaks signal. This yields the following high-utility scheme, which in particular sets $\bmu(S)=0$ for all subsets $S$ containing receiver $3$:
\begin{align*}
    \begin{cases}
        \bmu_G(\{1,2\})=1\\
        \bmu_B(\{1\})=0.5,\quad \bmu_B(\{1,2\})=0.5.
    \end{cases}
    \tag{Unpersuasive, but robust scheme}
\end{align*}

In general, similar masking ideas give an $\Omega(1/k)$ approximation to the optimal private utility when a randomly chosen subset of $k$ receivers broadcast their signal, a model we called ``$k$-broadcast''. We also consider three families of other leakage models and provide similar characterizations.

\subsection{Our Results}
We provide comprehensive answers to the above questions by studying two notions of robust private Bayesian persuasion in the presence of signal leakages.
The first notion is \emph{$k$-worst-case persuasiveness}, where the sender aims to design direct signaling schemes---where signals can be interpreted as action recommendations---that remain persuasive under arbitrary leakage patterns in which each receiver may observe up to $k$ leaked signals.
We let the optimal sender's utility in this setting be denoted by $\OPTpersuasive_k$.

The second notion is \emph{expected downstream utility robustness}, where we relax the requirement on persuasiveness and focus instead on achieving a high expected utility under specific parameterized distributions over leakage patterns.
In this context, ``downstream'' refers to the utility of the receivers’ final actions, which may differ from the action recommendations made by the mechanism.
Although direct signaling schemes are not always optimal (as discussed in \Cref{sec:discussion} and \Cref{sec:direct-vs-indirect}), we focus on them due to their simplicity and good performance, as we will show later. Here, each receiver best responds to both their own signal and
the leaked signals they observe;
this best response can be different from what their own signal originally suggests. For a distribution $\cG$ over leakage patterns, let $\OPTexpected(\cG)$ denote the best expected utility of an optimal scheme designed for $\cG$.
Using $\OPTprivate$ and $\OPTpublic$ to represent the optimal sender's utilities in the fully private and fully public settings, we have
\[
\OPTprivate\ge \OPTexpected(\cG)\ge\OPTpersuasive_k\ge\OPTpublic
\]
for any distribution $\cG$ over leakages patterns in which each receiver observes no more than $k$ leakages, as the first notion of worst-case persuasiveness is inherently stronger than the second notion.
Our objective is to bound the multiplicative gaps $\PoWR_k \coloneqq \OPTprivate/\OPTpersuasive_k$, called the \emph{Price of Robust Persuasiveness}, and $\PoDR(\cG)\coloneqq \OPTprivate/\OPTexpected(\cG)$, called the \emph{Price of Robust Downstream Utility}, which quantify the impact of leakages under the above two notions of robustness. We summarize our results in \Cref{table:results}.
\begin{table}[htbp]
    \renewcommand{\arraystretch}{1.6}
    \centering
    \footnotesize
    \begin{tabular}{|c|c|lr|lr|}
    \hline
    \multicolumn{2}{|c|}{Price of Robustness}  & \multicolumn{2}{c|}{Supermodular Utilities}  & \multicolumn{2}{c|}{XOS Utilities}\\
    \hline
    \multicolumn{2}{|c|}{
    Persuasiveness ($\PoWR_k$)
    }  & $\Theta(\min\{2^k,n\})$ &
    Thm~\ref{thm:gap-persuasive-general}, \ref{thm:upper-bound-private-public}, \ref{thm:lower-bound-persuasive-supermodular}
    & $\Theta(k)$ &Thm~\ref{thm:gap-persuasive-submodular}, \ref{thm:lower-bound-persuasive-submodular}\\
    \hline
    \multirow{4}{*}{\parbox[c][][c]{1.8cm}{Downstream
    Expected Utility ($\PoDR$)
    }}&$k$-Star & $\Theta(1)$ &Thm~\ref{thm:gap-expected-star}& $\Theta(1)$ &Prop~\ref{prop:expected-gap-star}\\ [.2em]
    \cline{2-6}
    &$k$-Clique & $\tilde{\Theta}(k)^{\dagger}$  &Thm~\ref{thm:gap-expected-broadcast}, \ref{thm:lower-bound-clique}& $\Theta(1)^{\mathsection}$&Prop~\ref{prop:expected-gap-clique} \\ [.2em]
    \cline{2-6}
    &$k$-Broadcast & $\Theta(k)^{\dagger}$ &Thm~\ref{thm:gap-expected-broadcast}, \ref{thm:lower-bound-broadcast}& $O(k)$&Thm~\ref{thm:gap-persuasive-submodular}\\ [.2em]
    \cline{2-6}
    &$k$-Erd\"os-R\'enyi & $O(\min\{2^k, n\})$  &Thm~\ref{thm:gap-persuasive-general}& $O(k)$ &Thm~\ref{thm:gap-persuasive-submodular}\\ [.2em]
    \hline
    \end{tabular}

    \caption{\footnotesize Bounds on the multiplicative gaps between $\OPTprivate$ and the optimal utilities subject to leakage-robustness. \\
    $^\dagger$represents that the lower bounds are proved only for prefix-based schemes (defined in \Cref{def:prefix-based}). $^\mathsection$means that the upper bound is based under the assumption that $k\le n-\Omega(n)$.
    $\tilde{\Theta}$ hides logarithmic factors in $k$.
    }
    \label{table:results}
\end{table}
\vspace{-2em}
\paragraph{Results for the $k$-worst-case persuasive setting.}
We view $k$-worst-case-persuasiveness as a gold standard of robust persuasion for two key reasons.
First by focusing on ``persuasiveness''---which only pertains to direct schemes---the directness of our scheme makes it simple and interpretable, both of which are broadly desirable in information design.
Second, the worst-case perspective ensures full robustness of the scheme across all possible leakages (up to $k$) and choices of utility.

To attain this gold standard, we establish upper bounds on $\PoWR_k$
by designing two generic approaches that transform privately persuasive schemes into $k$-worst-case persuasive ones through \emph{subsampling}.
For supermodular utility functions, our results demonstrate a ``phase transition'' in the robust benchmark $\OPTpersuasive_k$ at $k=\Theta(\log n)$: when $k\le\Theta(\log n)$, $\PoWR_k$ (the gap between $\OPTprivate$ and $\OPTpersuasive_k$) grows as $\Theta(2^k)$. However, there exist instances where, once $k\ge\Theta(\log n)$, $\PoWR_k$ plateaus and remains the same until $k=n-1$, where $\OPTpersuasive_{n-1}=\OPTpublic$. This phase transition indicates that private persuasion can tolerate no more than $k=\Theta(\log n)$ signal leakages before its effectiveness reduces to that of public persuasion.
These lower bounds indicate that the gold standard of robust persuasiveness should be viewed as largely out of reach---at least not without suffering significant utility loss on par with that of public persuasion.

\paragraph{Results for the expected downstream utility robustness setting.} To relax the worst-case perspective,
we consider four randomized leakage models $\cG$, each parametrized by $k$, the maximum number of leaked signals that can be observed by each receiver. In the \emph{$k$-star} model, $k$ out of the $n$ receivers chosen uniformly at random leak their signals to another receiver chosen uniformly at random from the remaining $n - k$. The \emph{$k$-clique} model assumes that $k$ uniformly random agents collude and share their signals with each other. In the \emph{$k$-broadcast} model, the $k$ uniformly random leaked signals are broadcast to the public and observed by all the $n$ receivers. Finally, in the \emph{$k$-Erd\"os-R\'enyi} model, each receiver observes a randomly and independently chosen subset of $k$ signals. We show that $\PoDR(\cG)$ (gap between $\OPTprivate$ and $\OPTexpected(\cG)$) improves upon $\PoWR_k$ when the leakage models are more structured and the sender is not required to be always persuasive.

Once again, our focus is on direct mechanisms. We view our results on expected downstream utility as largely positive and a testament to the effectiveness of direct schemes in preserving the utility of private persuasion under leakages. By using masking as an algorithmic technique, we show how a direct private scheme can be transformed---while maintaining its directness, interpretability, and simplicity---to obtain a $\PoDR=O(k)$ or even $O(1)$.

\subsection{Related Works}
\label{app:related}
Our work contributes to the research on algorithmic Bayesian persuasion (e.g.~\cite{dughmi2016algorithmic,dughmi2017algorithmic,babichenko2016computational,babichenko2017algorithmic,rubinstein2017honest,cheng2015mixture,bhaskar2016hardness,alonso2016persuading}), interpolations between private and public persuasion~\cite{multi-channel,mathevet2022organized,semi-public}, and robust Bayesian persuasion~\cite{de2022non,dworczak2022preparing,kosterina2022persuasion,hu2021robust,feng2024rationality,yang2024computational}. While signal leakages have not been extensively studied in prior works, we discuss the connection to studies on information spillover in networks~\cite{KT23,galperti2023games,egorov2020persuasion}.

\paragraph{Bayesian persuasion.} Our work builds on the foundational framework of Bayesian persuasion~\cite{KG2011bayesian}, and more specifically, multi-receiver Bayesian persuasion~\cite{bergemann2016bayes,bergemann2019information,taneva2019information,mathevet2020information}. This framework has broad applications, including voting~\cite{schnakenberg2015expert,AB19,alonso2016persuading,bardhi2018modes,wang2013bayesian},
bilateral trade~\cite{bergemann2015limits,bergemann2007information},
and security games~\cite{xu2015exploring,xu2016signaling},
to name a few.

One line of research on multi-receiver Bayesian persuasion has focused on \emph{private persuasion}~\cite{AB19,bergemann2016bayes,kerman2024persuading,babichenko2016computational,babichenko2017algorithmic,alg-no-externality}, where the sender communicates with receivers through private communication channels. These work reveal that while computing the optimal private signaling scheme is tractable in certain cases without inter-receiver externalities, it becomes computationally intractable in presence of inter-receiver externalities even for zero-sum games.
Another line of research has focused on \emph{public persuasion}~\cite{public-almost-optimal,bhaskar2016hardness,rubinstein2017honest,dughmi2017algorithmic,xu2020tractability,alg-no-externality,cheng2015mixture,candogan2019persuasion} where the sender is constrained to using public communication channels. Prior work \cite{nachbar2022power,alg-no-externality} establishes that private persuasion can be much more powerful than public persuasion with the gap as high as the Price of Anarchy of the underlying game.

\paragraph{Settings between private and public persuasion.}
Our work introduces a model of private persuasion with leakages that interpolates between private and public persuasion. Prior works have considered other types of communication channels between private and public ones, for example, semi-public channels~\cite{semi-public} in the context of district-based elections where the sender can use a single communication per district; multi-channel persuasion~\cite{multi-channel} where each receiver observes a subset of the sender's communication channels, and organized information transmission including horizontal and vertical information structures~\cite{mathevet2022organized}.

\paragraph{Information spillover in networks.}
Our work is closely related to the research on \emph{information spillover} in networks.
The work of \cite{KT23} studies a setting similar to ours, where agents in a network observe both their own signal and the signals received by their neighbors. They highlight the non-monotonicity of the sender's utility with network density, and characterize a class of networks where monotonicity holds. The main difference between our work and theirs is that they assume the sender knows the graph structure, whereas we consider robustness perspectives where the graph is unknown a priori and chosen either adversarially or stochastically after the sender commits to a signaling scheme.
\cite{galperti2023games} studies {network-seed systems}, {in which} information diffuses along all directed paths in the network, and the sender provides information to a subset of agents referred to as ``seeds''.
\cite{egorov2020persuasion} studies a setting where the sender can only communicate publicly with receivers, and each receiver either relies on their neighbors or learn directly from the center at a cost.

\paragraph{Robustness of Bayesian persuasion.}
Our work also contributes to the literature on the robustness of Bayesian persuasion. Most prior works study robustness under the imperfectness or irrationality of receivers~\cite{feng2024rationality,de2022non,chen2023persuading,yang2024computational}, or assume the sender has limited knowledge about the receivers or the environment~\cite{dworczak2022preparing,kosterina2022persuasion,hu2021robust,babichenko2022regret,camara2020mechanisms,collina2023efficient}. Our work differs from both perspectives by focusing on the robustness of persuasion when the communication channels are imperfect and subject to signal leakages.

\section{Model and preliminaries}
\label{sec:model}
\subsection{Basics of Private Bayesian Persuasion}
We start by introducing the basic setting of \emph{private Bayesian persuasion} without leakages from \cite{AB19}. Let $\Omega=\{\omega_0,\omega_1\}$ be a binary state space and $\prior\in\Delta(\Omega)$ be a common prior distribution, in which $\omega_1$ has probability $\lambda$ and $\omega_0$ has probability $1-\lambda$. Consider a single sender and a group of $n$ receivers (which we also refer to as \emph{agents}), denoted by $N=[n]=\{1,2,\ldots,n\}$. Each receiver $i\in N$ has a binary action space $\A=\{0,1\}$, and will adopt action $1$ if and only if they believe that the probability of the state $\omega_0$ is no greater than a threshold $p_i$,\footnote{\label{footnote:threshold}This thresholding behavior can also be viewed as best response based on utility function $u_i:\Omega\times \A\to \R$. As in \cite{AB19}, if each receiver strictly prefers the action that matches the true state, i.e., $\forall i$, $u_i(\omega_0,0)>u_i(\omega_0,1)$ and $u_i(\omega_1,1)>u_i(\omega_1,0)$, we have $p_i=\frac{u_i(\omega_0,0)-u_i(\omega_0,1)}{u_i(\omega_1,1)-u_i(\omega_1,0)+u_i(\omega_0,0)-u_i(\omega_0,1)}$.} i.e., $a_i=\1{\Pr[\omega_0]\le p_i}$. The sender's utility $V: 2^N \to \R$ is a function of the group of receivers that take action $1$ (also called \emph{adopters}). We assume that $V$ is monotone, i.e., $V(S)\le V(T)$ for all $S\subseteq T$. Without loss of generality, we also assume $V(\emptyset)=0$. When clear from the context, we sometimes abuse the notation and write $V$ as a function of all receivers' actions, i.e., $V(a_1,\ldots,a_n)=V(\{i\in N\mid a_i=1\})$.

The sender commits to a signaling scheme that consists of a signal space $\signalspace=\signalspace_1\times\ldots\times\signalspace_n$ and signal distributions $\bmuzero,\bmuone\in\Delta(\signalspace)$ to be used conditional on state $\omega_0$ and $\omega_1$, respectively. After the random state is drawn from the prior distribution $\omega\sim\prior$, the sender observes $\omega$ and generates signals $(s_1,\ldots,s_n)\sim\bmuomega$ according to the committed scheme. Each receiver $i$ then receives their signal $s_i$ from a private communication channel.

In this paper, we mainly consider \emph{direct} signaling schemes where $\signalspace_i=\A=\{0,1\}$ for every receiver $i \in N$, so that signals can be interpreted as recommendations on whether to adopt. This restriction is without loss of generality in private persuasion due to the revelation principle (e.g., see~\cite{bergemann2016bayes}).\footnote{In the presence of signal leakages, restricting to direct signaling schemes is no longer without loss of generality. We discuss this in \Cref{sec:direct-vs-indirect}.} For a direct signaling scheme $(\bmuzero, \bmuone)$, each signal realization $(s_1,\ldots,s_n)\in\{0,1\}^n$ can be interpreted as a subset of adopters $S=\{i \in N\mid s_i=1\}$, allowing us to sometimes also view $\bmuzero$ and $\bmuone$ as distributions over subsets of $N$.

We start with the fully private setting where each receiver $i$ can only observe their own signal $s_i$. By Bayes' rule, after receiving $s_i$, receiver $i$'s posterior belief can be computed as \[\pr{}{\omega_0\mid s_i}=\frac{\pr{}{\omega_0}\cdot\pr{}{s_i\mid \omega_0}}{\pr{}{s_i}}=\frac{(1-\lambda)\sum_{s_{-i}}\bmuzero(s_i,s_{-i})}{(1-\lambda)\sum_{s_{-i}}\bmuzero(s_i,s_{-i})+\lambda\sum_{s_{-i}}\bmuone(s_i,s_{-i})}.\]
We use $a_i^{\bmu}(s_i)$ to denote receiver $i$'s best response to signal $s_i$ under signaling scheme $\bmu$. When the signaling scheme is clear from the context, we may omit the superscript and simply write $a_i(s_i)$. Since receiver $i$ adopts action $1$ if and only if the posterior probability of $\omega_0$ is no more than $p_i$, the best response function is given by:
\begin{align}
    a_i^{\bmu}(s_i)=\1{\pr{\prior,\bmu}{\omega_0\mid s_i}\le p_i}=\1{
\sum_{s_{-i}}\bmuzero(s_i,s_{-i})\le\frac{\lambda}{1-\lambda}\cdot\frac{p_i}{1-p_i}\cdot\sum_{s_{-i}}\bmuone(s_i,s_{-i})}.
\label{eq:br-without-leakage}
\end{align}
Following \citet{AB19}, we define the \emph{persuasion level} of receiver $i$ as
\begin{align}
    \theta_i:=\frac{\lambda}{1-\lambda}\cdot\frac{p_i}{1-p_i}.
    \label{eq:persuasion-level-def}
\end{align}
Without loss of generality, we assume $1\ge \theta_1\ge\theta_2\ge\ldots\ge\theta_n\ge0$. Intuitively, the persuasion level reflects how easily a receiver can be persuaded to take action $1$. A higher $\theta_i$ indicates that the receiver has higher tolerance to the unfavorable state $\omega_0$ and is thus more easily persuaded.  For example, in the online advertising setting, a higher $\theta_i$ means that a customer is less demanding about the product quality and can be more easily persuaded to purchase.

A signaling scheme $\bmu$ is \emph{privately persuasive} if $a_i^{\bmu}(s_i)=s_i$ for all $i\in N$ and $s_i\in\{0,1\}$.\footnote{When a receiver is indifferent between two actions, which happens when the posterior probability $\pr{\prior,\bmu}{\omega_0\mid s_i}$ is exactly $p_i$, we abuse the notation of $a_i^{\bmu}(s_i)=s_i$ and allow $s_i$ to be either $0$ or $1$.} We use $\OPTprivate$ to denote the sender utility of the optimal privately persuasive signaling scheme:
\[
\OPTprivate:=\sup_{\bmu \text{ privately persuasive}}\left\{
\Ex{\omega\sim\prior\atop (s_1,\ldots,s_n)\sim\bmuomega}{V(s_1,\ldots,s_n)}
\right\}.
\]

\subsection{Signal Leakages}
We extend the basic setting of private Bayesian persuasion and examine its robustness in the presence of signal leakages.
More specifically, after each agent receives their signal $s_i$, we assume that a \emph{leakage pattern} is either adversarially or stochastically generated, according to which signals are leaked to other receivers. Formally, a leakage pattern is defined as a directed graph $G=(N,E)$ of communication channels between agents, where each directed edge $e=(u\to v)\in E$ represents that agent $u$ leaks their signal $s_u$ to agent $v$. As a result, each agent $i\in N$ observes not only their own signal $s_i$, but the signals of all their in-neighbors in graph $G$. We use $I_i^G:=\{(j,s_j)\mid (j\to i)\in E\}$ to denote the additional information available to receiver $i$ after the leakages, which consists of both the observable signal values and the identities of the agents who leaked them. When clear from the context, we will abbreviate it as $I_i$. Note that our model can be viewed as a generalization of the private Bayesian persuasion model of~\citet{AB19}, which corresponds to the special case of $I_i=\emptyset$.

After receiving additional information from leakages, each receiver $i$ updates their posterior belief based on both $I_i$ and $s_i$. Using $I_i\triangleright s_{-i}$ as a shorthand to indicate that $I_i$ is consistent with $s_{-i}$ (i.e., all the leaked signals in $I_i$ match the corresponding components in $s_{-i}$), the posterior distribution can be calculated as:
\[\pr{}{\omega_0\mid s_i,I_i}=\frac{\pr{}{\omega_0}\pr{}{s_i, I_i\mid \omega_0}}{\pr{}{s_i,I_i}}=\frac{(1-\lambda)\sum_{s_{-i}:I_i\triangleright s_{-i}}\bmuzero(s_i,s_{-i})}{(1-\lambda)\sum_{s_{-i}:I_i\triangleright s_{-i}}\bmuzero(s_i,s_{-i})+\lambda\sum_{s_{-i}:I_i\triangleright s_{-i}}\bmuone(s_i,s_{-i})}.\]
In other words, when computing the posterior, each agent now considers all possible completions of the leaked signals they observe, and weights them according to their likelihood under different states.
We use $a_i^{\bmu}(s_i,I_{i})$ to denote the best response to both their private signal $s_i$ and the additional leaked information $I_i$  under signaling scheme $\bmu$. As before, when the signaling scheme is clear from the context, we may omit the superscript $\bmu$ and simply write $a_i(s_i,I_{i})$. The best response function can then be written as:
\begin{align}
a_i^{\bmu}(s_i,I_{i})=\1{\sum_{s_{-i}:I_i\triangleright s_{-i}}\bmuzero(s_i,s_{-i})\le\theta_i\cdot\sum_{s_{-i}:I_i\triangleright s_{-i}}\bmuone(s_i,s_{-i})}.
\label{eq:br-leakage}
\end{align}

 \subsection{Notions of Leakage-Robustness}
We consider two notions of robustness in the presence of signal leakages, \emph{worst-case persuasiveness} and \emph{expected downstream utility robustness}.

\paragraph{$k$-worst-case persuasiveness.} In this notion, a signaling scheme is required to not only be persuasive in the basic private setting, but also ensure that the set of adopter does not shrink under any leakage patterns with in-degree at most $k$. Specifically, this means that even when each agent observes up to $k$ additional leaked signals, adopters who initially chose action $1$ before the leakages must continue to do so afterward.\footnote{Note that we do not require that agents who choose action $0$ before leakages to not deviate to $1$ afterward. As we show in \Cref{prop:two-sided-persuasiveness}, only public schemes can satisfy this stronger, two-sided version of $k$-worst-case persuasiveness once $k\ge2$.}
We focus on preventing the set of adopters from shrinking because the sender’s utility is monotone in the set of adopters, so ensuring the set does not shrink preserves or improves the sender's utility.

\begin{definition}[$k$-Worst-Case Persuasiveness]
\label[definition]{def:k-worst-case-persuasiveness}
A direct signaling scheme $\bmu$ is \emph{$k$-worst-case persuasive} if:
    \begin{enumerate}
        \item[(1)] It is privately persuasive: $\forall i\in N, s_i\in\{0,1\}$, the best response is  $a_i^{\bmu}(s_i)=s_i$;
        \item[(2)] It guarantees that agents receiving signal value $1$ would not deviate to action $0$ under leakages, as long as each agent observes no more than $k$ leaked signals: $\forall i\in N$ and $I_i$ such that $|I_i|\le k$, $a_i^{\bmu}(s_i, I_i) = s_i$ holds for $s_i = 1$.
    \end{enumerate}
\end{definition}

The sender's goal is to choose the optimal $k$-worst-case persuasive signaling scheme that maximizes their utility in the private setting. We formally define this optimal sender utility as:\footnote{Note that $\OPTpersuasive_k$ considers the sender's utility in the fully private setting without leakages. This characterization is equivalent to the maxmin utility for worst-case leakage patterns as shown in \Cref{prop:maxmin-utility}.}
\[
\OPTpersuasive_k:=\sup_{\bmu \text{ $k$-worst-case persuasive}}\left\{
\Ex{\omega\sim\prior\atop (s_1,\ldots,s_n)\sim\bmuomega}{V(s_1,\ldots,s_n)}
\right\}.
\]
When $k=n-1$, each agent can observe the signals of all other agents, effectively making all signals public. Therefore, we call a signaling scheme \emph{publicly persuasive} if it is $(n-1)$-worst-case persuasive, and define $\OPTpublic \coloneqq \OPTpersuasive_{n-1}$.

\paragraph{Expected downstream utility robustness.} This notion relaxes the requirement of strict persuasiveness and instead focuses on the downstream impact of leakages in the sender's utility. In this notion, we assume that the leakage pattern $G$ is independently sampled from a distribution $\cG$ of leakage patterns (i.e., $G\sim\cG$), and we evaluate the expectation of the effective sender utility after all receivers best respond to the leakages. Importantly, we no longer require that the receivers' best responses are consistent with the signals that they initially receive. Formally, for a distribution $\cG$ over potential leakage patterns, the sender's optimal expected downstream utility, denoted by $\OPTexpected(\cG)$, is defined as
\[
\OPTexpected(\cG):=\sup_{\bmu}\left\{\Ex{\omega\sim \prior\atop (s_1,\ldots,s_n)\sim\bmuomega
}{
\Ex{G\sim\cG}{V(a_1^{\bmu}(s_1,I_1^G), a_2^{\bmu}(s_2,I_2^G), \ldots, a_n^{\bmu}(s_n,I_n^G))}
}\right\}.
\]
In this paper, we focus on four specific choices of $\cG$, each of which is parameterized by $k$, the maximum in-degree of all graphs in the family. The distribution $\cG$ can be equivalently described by the process through which a graph $G\sim\cG$ is sampled:
\begin{itemize}
    \item \textbf{$k$-Star:} $G$ contains $\{(j_1\to i), ( j_2\to i), \ldots, (j_k\to i)\}$ for $k + 1$ different agents $i, j_1, \ldots, j_k \in N$ chosen uniformly at random.
    \item \textbf{$k$-Clique:} $G$ contains all directed edges among $k$ different agents $j_1, \ldots, j_k\in N$ chosen uniformly at random.
    \item \textbf{$k$-Broadcast:} A size-$k$ subset $S \subseteq N$ is drawn uniformly at random. $G$ contains all the edges $\{(j \to i) \mid i \in N, j \in S, i \ne j\}$.
    \item \textbf{$k$-Erd\"os-R\'enyi:} For each $i \in N$, a size-$k$ subset $S_i \subseteq N \setminus \{i\}$ is drawn uniformly and independently at random. $G$ contains all the edges $\{(j \to i)\mid i \in N, j\in S_i\}$.
\end{itemize}

\paragraph{Comparison of Robustness Benchmarks.}
As discussed in the introduction, it is not hard to see that $\OPTexpected(\cG)\ge\OPTpersuasive_k$ for $\cG\in\{\kstar,\kclique,\kbroadcast,\kErdosRenyi\}$ since each of the four choices for $\cG$ has a maximum in-degree of at most $k$. Overall, we have
\[
\OPTprivate\ge \OPTexpected(\cG)\ge\OPTpersuasive_k\ge\OPTpublic.
\]

\section{Leakage-Robust Schemes via Subsampling and Masking}
\label{sec:upper-bound}

\subsection{Design of Worst-Case Persuasive Schemes via Subsampling}
\label{sec:technical-subsampling}
In this section, we consider the design of $k$-worst-case persuasive signaling schemes. We start with an observation that $\OPTpersuasive_k$ is characterized by the solution to the following linear program (LP) with variables $\bmuzero$ and $\bmuone$ indexed by $\{0, 1\}^n$:
\begin{align}
    \operatorname*{Maximize}_{\bmuzero,\bmuone}&\sum_{(s_1,\ldots,s_n)\in\{0,1\}^n}
    \left(\lambda \bmuone(s_1,\ldots,s_n)+(1-\lambda)\bmuzero(s_1,\ldots,s_n)\right)\cdot V(s_1,\ldots,s_n)
    \tag{LP}
    \label{eq:lp-persuasive}\\
    \subjectto&
    \circled{1}\ \forall i\in[n],\ \begin{cases}\sum_{s_{-i}}\bmuzero(1,s_{-i})\le\theta_i\cdot\sum_{s_{-i}}\bmuone(1,s_{-i})\\
    \sum_{s_{-i}}\bmuzero(0,s_{-i})\ge\theta_i\cdot\sum_{s_{-i}}\bmuone(0,s_{-i})
    \end{cases}\nonumber\\
    &\circled{2}\ \forall i\in[n],\ \forall |I_i|\le k,\ \sum_{s_{-i}:I_i\triangleright s_{-i}}\bmuzero(1,s_{-i})\le\theta_i\cdot\sum_{s_{-i}:I_i\triangleright s_{-i}}\bmuone(1,s_{-i})\nonumber\\
    &\circled{3}\ \bmuzero,\bmuone\ge0,\ \sum_{s}\bmuzero(s_1,\ldots,s_n)=1,\  \sum_{s}\bmuone(s_1,\ldots,s_n)=1\nonumber
\end{align}
In the above LP, condition \circled{1} ensures private persuasiveness by enforcing $a_i(s_i)=s_i$ using the best response function defined in \Cref{eq:br-without-leakage}. Condition \circled{2} ensures that no adopters deviate after observing up to $k$ \emph{arbitrary} leaked signals, based on the best response function $a_i(s_i,I_i)$ defined in \Cref{eq:br-leakage}. Finally, condition \circled{3} requires $\bmuzero$ and $\bmuone$ to be valid probability distributions.

We begin with a proposition that highlights the non-robustness of optimal private schemes even when a single signal leakage occurs. The underlying reason for this non-robustness also provides insights into the design of worst-case persuasive schemes.

\begin{restatable}{proposition}{PrivateNonrobust}
\label[proposition]{prop:opt-nonrobust}
    When the persuasion levels $\theta_1, \ldots, \theta_n$ are not identical, the optimal private scheme from \cite{AB19} fails to be $1$-worst-case persuasive.
\end{restatable}

\begin{proof}[Proof of \Cref{prop:opt-nonrobust}]
    The optimal solution achieving $\OPTprivate$ is characterized by the linear program in \eqref{eq:lp-persuasive} with condition \circled{2} removed
    (which is consistent with the LP described in \cite[Corollary 1]{AB19}). We denote this optimal solution with $(\bmuzero^\star, \bmuone^\star)$.
    From \cite{AB19},
    $\bmuone^\star$ concentrates all the probability mass on the full set $N$, meaning all agents receive signal $1$ under state $\omega_1$.
    For each $0\le j\le n$, $\bmuzero^\star$ assigns $\theta_j-\theta_{j+1}$ probability to the prefix $[j]=\{1,\ldots,j\}$ (with $\theta_0 = 1$ and $\theta_{n+1} = 0$). Since $\theta_1$ through $\theta_n$ are not identical, there exists a prefix $[j]\neq N$ that receives a nonzero probability.

    However, since $\bmuone^\star$ concentrates all the probability mass on the full set $N$, when the prefix $[j]$ is realized, a single leakage of signal $0$ from any agent in $N\setminus[j]$ would immediately reveal that the true state is $\omega_0$ and cause adopters who observe it to deviate.
    More specifically, consider an adopter $i\in [j]$ who observes a signal $0$ from agent $j + 1$. Since $\bmuone^\star$ assigns a zero probability to every proper subset of $N$, the right-hand side of \circled{2} is $0$.
    In contrast, the left-hand side of \circled{2} remains positive, since $[j]$ is consistent with the leaked signal and assigned a positive probability under $\bmuzero^\star$. This discrepancy results in a violation of condition $\circled{2}$. This implies that the privately optimal scheme $(\bmuzero^\star, \bmuone^\star)$ fails to be even 1-worst-case persuasive.
\end{proof}

In light of~\Cref{prop:opt-nonrobust}, we design worst-case persuasive schemes by introducing randomness into $\bmuone^\star$, thus preventing the leaked signals of value $0$ from fully revealing the state. Rather than sending $1$ to all receivers in $N$ under state $\omega_1$, we construct $\bmuone$ by sending $1$ to a \emph{subsampled} subset of receivers where each receiver is independently included with some probability $\gamma$. This independent subsampling ensures that the right-hand side of condition \circled{2} remains positive under any worst-case leakage pattern, as opposed to in the optimal private scheme $(\bmuzero^\star, \bmuone^\star)$. In the following, we propose two approaches that transform $(\bmuzero^\star, \bmuone^\star)$ into $k$-worst-case persuasive schemes that satisfy condition \circled{2}, each paired with a different subsampling rate $\gamma$ for $\bmuone$.

In the first approach, we use a subsampling rate of $\gamma=\frac{1}{2}$, as shown in the following lemma:
\begin{restatable}[Subsampling at rate $\frac{1}{2}$]{lemma}{LemSubsampleHalf}
\label[lemma]{lemma:subsample-1-2}
    Let $(\bmuzero^\star,\bmuone^\star)$ be a privately persuasive signaling scheme with $\bmuone^\star(1,\ldots,1)=1$. The following scheme $(\bmuzero,\bmuone)$ is $k$-worst-case persuasive:
    \begin{itemize}
        \item For all $(s_1,\ldots,s_n)\in\{0,1\}^n$, $\bmuone(s_1,\ldots,s_n)=2^{-n}$;
        \item For all $(s_1,\ldots,s_n)\neq(0,\ldots,0)$, $\bmuzero(s_1,\ldots,s_n)=2^{-(k+1)}\cdot\bmuzero^\star(s_1,\ldots,s_n)$.
    \end{itemize}
\end{restatable}

We explain the reasoning behind selecting $\gamma=\frac{1}{2}$ and outline a proof sketch of \cref{lemma:subsample-1-2}.
Our goal is to satisfy \circled{2} by uniformly lower bounding the right-hand across all leakage patterns $I_i$, while also upper bounding the left-hand side by reducing the probability assigned to all nonempty subsets of adopters in $\bmuzero$.
Concretely, when $\bmuone$ includes each receiver independently with rate $\gamma$, the probability of observing both $s_i=1$ and a specific leakage pattern $I_i$ is given by
\begin{align}
    \sum_{s_{-i}:I_i\triangleright s_{-i}}\bmuone(1,s_{-i})
=\gamma^{|I_i^+|+1}\cdot(1-\gamma)^{|I_i^-|},
\label{eq:cond-2-rhs}
\end{align}
where $I_i^+$ (resp., $I_i^-$) represents the leaked signals of value $1$ (resp., $0$).
By choosing $\gamma=\frac{1}{2}$, we obtain a good uniform lower bound for this probability across all possible leakage patterns,
as the right-hand side of $\eqref{eq:cond-2-rhs}$ is at least $2^{-(k+1)}$ for every leakage pattern $I_i$ with $|I_i|\le k$.

As a result, to satisfy condition \circled{2},
it suffices to upper bound the left-hand side by rescaling the probability that $\bmuzero^\star$ assigns to each nonempty adopter set by a factor of $2^{-(k+1)}$.
This allows $(\bmuzero,\bmuone)$ to inherit from the private persuasiveness of $(\bmuzero^\star,\bmuone^\star)$:
\begin{align*}
    \sum_{s_{-i}:I_i\triangleright s_{-i}}\bmuzero(1,s_{-i})\le
    \sum_{s_{-i}}\bmuzero(1,s_{-i})= 2^{-(k+1)}\sum_{s_{-i}}\bmuzero^\star(1,s_{-i})
    \le 2^{-(k+1)}\cdot\theta_i
    = \theta_i \sum_{s_{-i}:I_i\triangleright s_{-i}}\bmuone(1,s_{-i}).
\end{align*}
We have thus satisfied condition \circled{2}. To complete the proof of \cref{lemma:subsample-1-2}, it remains to verify that the resulting scheme $(\bmuzero, \bmuone)$ also satisfies condition \circled{1}, which we address in \Cref{app:proof-subsample-1-2}.

In the following theorem, we apply \Cref{lemma:subsample-1-2} to establish an $O(2^k)$ bound on $\PoWR_k$.
\begin{restatable}[{General sender utilities}]{theorem}{gapPersuasiveGeneral}
    \label{thm:gap-persuasive-general} For any $k$, $\PoWR_k\le O(2^k)$.
\end{restatable}
\begin{proof}[Proof of \Cref{thm:gap-persuasive-general}]
    We prove this theorem by constructing a $k$-worst-case persuasive scheme with a utility of $\Omega(2^{-k})\cdot\OPTprivate$. Let $(\bmuzero^\star,\bmuone^\star)$ be the optimal private scheme that achieves $\OPTprivate$. Note that it is without loss of generality to assume $\bmuone^\star(N)=1$, otherwise we can replace $\bmuone^\star$ with the point distribution on $N$ to obtain a persuasive scheme in which the sender's utility does not decrease. We can write $\OPTprivate= \OPTprivate(\omega_1)+\OPTprivate(\omega_0)$, where
    $\OPTprivate(\omega_1)=\lambda\cdot V(1,\ldots,1)$ is the sender's utility under state $\omega_1 $, and \[
    \OPTprivate(\omega_0)=(1-\lambda)\sum_{(s_1,\ldots,s_n)\in\{0,1\}^n} \bmuzero^\star(s_1,\ldots,s_n)\cdot V(s_1,\ldots,s_n)\]
    is the utility under $\omega_0$.
    We consider two cases:
    \begin{itemize}
        \item If $ \OPTprivate(\omega_1)\ge\frac{1}{2} \OPTprivate$, we can use a full-information scheme that deterministically sends $(1,\ldots,1)$ under state $\omega_1$ and $(0,\ldots,0)$ under state $\omega_0$. This full information scheme is $k$-worst-case persuasive and achieves utility $ \OPTprivate(\omega_1)\ge\frac{1}{2} \OPTprivate$.
        \item If $ \OPTprivate(\omega_0)\ge\frac{1}{2} \OPTprivate$, we apply \Cref{lemma:subsample-1-2} to $(\bmuzero^\star,\bmuone^\star)$ and obtain a $k$-worst-case persuasive scheme $(\bmuzero, \bmuone)$.
        Since $\bmuzero(S)=2^{-(k+1)}\bmuzero^\star(S)$ for all nonempty adopter subsets $S$, the utility of this scheme is at least
        \[
        (1-\lambda)\sum_{s \in \{0,1\}^n} 2^{-(k+1)}\cdot\bmuzero^\star(s)\cdot V(s)
        =2^{-(k+1)}\cdot\OPTprivate(\omega_0)
        \ge 2^{-(k+2)}\cdot\OPTprivate.
        \]
    \end{itemize}
    In either case, we obtain a $k$-worst-case persuasive scheme of utility $\Omega(2^{-k})\cdot\OPTprivate$. This establishes the upper bound $\PoWR_k=\OPTprivate/\OPTpersuasive_k\le O(2^k)$.
\end{proof}

While this $2^k$ dependence seems daunting, it cannot be improved in general: In \Cref{sec:technical-lower-bound},
we construct a hard instance on which $\PoWR_k$ is $\Omega(2^k)$. The hard instance has a supermodular sender utility function under which shrinking the adopter set by even a single agent may result in a complete loss of the utility.

Fortunately, the price of robust persuasiveness can be significantly reduced when the instance has a more favorable utility structure. In the remainder of this section, we present a second approach (\Cref{lemma:subsample-1-k} below) that, {when applied at a rate of $\Theta(1/k)$,} potentially reduces the gap from $O(2^k)$ to $O(k)$. This result relies on the assumption that taking random subsets of adopters does not significantly reduce the sender's utility. Specifically, this applies when
\begin{align}
    \Ex{S'\sim \D_\gamma(S)}{V(S')}\ge\gamma\cdot V(S),
    \label{eq:structured-utility}
\end{align} where $\D_\gamma(S)$ represents the distribution of $S'$ obtained by including each receiver in $S$ independently with probably $\gamma$. This condition holds when $V$ is submodular or XOS, as shown in \Cref{app:submodular-subadditive-property}.

\begin{restatable}[Subsampling at a general rate]{lemma}{LemSubsampleK}
    \label[lemma]{lemma:subsample-1-k}
    Let $(\bmuzero^\star,\bmuone^\star)$ be privately persuasive with $\bmuone^\star(1,\ldots,1)=1$, and $\gamma\in[0,1]$ be the subsampling rate. The following scheme $(\bmuzero,\bmuone)$ is $k$-worst-case persuasive:
    \begin{itemize}
        \item $\bmuone$ is defined by independently sending signal $1$ to each agent with probability $\gamma$ and $0$ with probability $1-\gamma$.
        Equivalently, $\bmuone(s_1,\ldots,s_n)=\gamma^{l}(1-\gamma)^{n-l}$, where $l=\sum_{i=1}^{n} s_i$.
        \item $\bmuzero$ is defined by the following procedure for drawing a random subset $S' \subseteq N$: With probability $1 - (1 - \gamma)^k$, set $S' \gets \emptyset$. With the remaining probability $(1 - \gamma)^k$, first draw $S \sim \bmuzero^\star$, and then subsample $S$ at rate $\gamma$ to form a new set $S' \sim \D_\gamma(S)$, i.e., each $i \in S$ is included in $S'$ independently with probability $\gamma$. Signal $1$ is sent to all agents in $S'$.
    \end{itemize}
\end{restatable}

Before diving into the proof, we show how Lemma~\ref{lemma:subsample-1-k} can be applied to prove an improved upper bound on the multiplicative gap between $\OPTprivate$ and $\OPTpersuasive_k$ for structured utility functions.
\begin{restatable}[{Structured sender utilities}]{theorem}{gapPersuasiveSubmodular}
\label[theorem]{thm:gap-persuasive-submodular}
    For any $k$, if the sender's utility function satisfies \Cref{eq:structured-utility}, we have $\PoWR_k\le O(k)$.
\end{restatable}
\begin{proof}[Proof of \Cref{thm:gap-persuasive-submodular}]
    Similar to the proof of \Cref{thm:gap-persuasive-general}, it suffices to construct a signaling scheme that achieves an $\Omega(1/k)$ fraction of the optimal sender utility under state $\omega_0$. This is accomplished by applying \Cref{lemma:subsample-1-k} to the optimal private scheme $(\bmuzero^\star,\bmuone^\star)$ {at rate $\gamma=1/k$} and obtain a new scheme $(\bmuzero,\bmuone)$. By \cref{lemma:subsample-1-k}, $(\bmuzero,\bmuone)$ is $k$-worst-case persuasive, and the sender's utility under state $\omega_0$ is given by
    \begin{align*}
        &~(1 - \lambda)\cdot\sum_{S\in 2^N}(1-\gamma)^k\cdot\bmuzero^\star(S)\Ex{S'\sim\D_\gamma(S)}{V(S')}\\
        \ge&~(1 - \lambda)\cdot\sum_{S\in 2^N}(1-\gamma)^k \gamma\cdot\bmuzero^\star(S)\cdot V(S)
        \tag{$V$ satisfies \Cref{eq:structured-utility}}\\
        \ge&~(1-\gamma)^k \gamma\cdot \OPTprivate(\omega_0)
        =~\Omega(1/k)\cdot\OPTprivate(\omega_0).
        \tag{$\gamma=1/k$}
    \end{align*}
    Thus, the better of $(\bmuzero,\bmuone)$ and a full-information scheme achieves $\OPTpersuasive_k\ge\Omega(\frac{1}{k})\cdot\OPTprivate$, which implies $\PoWR_k\le O(k)$.
\end{proof}

We now turn to the proof sketch of \Cref{lemma:subsample-1-k}.
{By the construction of $\bmuzero$,} for agent $i$ to observe a positive signal $s_i=1$ along with the positive leaked signals $I_i^+$, it is necessary to first sample a set $S\sim\bmuzero^\star$ such that $S\supseteq (I_i^+\cup\{i\})$, and then to ensure that every element in $I_i^+ \cup \{i\}$ is included in $S'$ after subsampling. This happens with probability at most $\gamma^{|I_i^+|+1}$. Comparing this with the probability of observing the same leaked signals $I_i=I_i^+\cup I_i^-$ under $\bmuone$ (derived in \Cref{eq:cond-2-rhs}), we see that the factor $\gamma^{|I_i^+|+1}$ appears on both sides of condition \circled{2}. Therefore, to satisfy \circled{2}, it suffices to set the rescaling factor of $\bmuone$ to match $(1-\gamma)^{|I_i^-|}\ge (1-\gamma)^k$.
{Thus, by introducing subsampling to $\bmuzero$, we bypass the need to uniformly scale $\bmuone$ by an unfavorable factor of $\Theta(2^{-k})$ (as was necessary in \Cref{lemma:subsample-1-2}) and instead apply a more favorable scaling factor that significantly improves the sender's utility.}
We formally prove \Cref{lemma:subsample-1-k} in \Cref{app:proof-subsample-1-k}.

\subsection{Design of Expected Downstream Utility Robust Schemes via Masking}
\label{sec:technical-expected-util-robust}

In this section, we design signaling schemes in the expected {downstream} utility robustness model. Without enforcing persuasiveness under all leakage patterns, the sender aims to achieve a high expected utility after each receiver best responds to leaked signals under specific distributions of leakage patterns.

The high-level idea behind our signaling scheme design is simple: we identify the ``typical'' combinations of leakage patterns and signal realizations (under the optimal private scheme) that occur with high probability, and adjust the schemes to preserve the adopters under these combinations.
This is often achieved by \emph{``masking out''} the information carried by signals that may be leaked.
In the following, we discuss two masking approaches used to develop high-utility schemes.

\paragraph{\textbf{Masking by removing randomness.}}
The first approach performs \emph{explicit masking} by removing the randomness of the signals that might be leaked, and replacing such signals with deterministic ones that carry no information. This approach is applicable to any leakage models where $k$ randomly chosen signals are leaked---including the $k$-star, $k$-clique and $k$-broadcast models---and supermodular sender utilities under which the optimal private solutions have favorable structures.

Before introducing our approaches, we first review the structure of the optimal private scheme $(\bmuzero^\star,\bmuone^\star)$ derived by~\citet{AB19}. As discussed in the previous section, $\bmuone^\star$ concentrates all the probability mass on the full set $N$. As for $\bmuzero^\star$, it assigns positive probabilities only to the empty set and ``prefix sets'' of the form $[i]=\{1,2,\ldots,i\}$, with each prefix $[i]$ receiving probability $\theta_i-\theta_{i+1}$ (recall that the persuasion levels are sorted as $\theta_1\ge\theta_2\ge\ldots\ge\theta_n$; also assume $\theta_{n+1}=0$). In other words, under state $\omega_0$, the sender draws a prefix $[i]$ from the above distribution and send signal $1$ only to agents $i'$ with indices $i'\le i$. We state our first masking approach in \Cref{lemma:masking-removing-randomness}.

\begin{restatable}[Masking by removing randomness]{lemma}{maskRemovingRandomness}
\label[lemma]{lemma:masking-removing-randomness}
Suppose that the sender's utility is supermodular, and the distribution of leakage patterns satisfies that a randomly chosen subset of $k \le \frac{n}{2}$ agents leak their signals. For any $i\le n-\lfloor\frac{n}{k}\rfloor$, the following scheme $\left(\bmuzero^{(i)},\bmuone^{(i)}\right)$ achieves an expected utility of at least $\Omega(1)\cdot\sum_{j=i}^{i+\lfloor n/k\rfloor}(1-\lambda)\cdot(\theta_j-\theta_{j+1})\cdot V([j]).$
\begin{itemize}
    \item $\bmuone^{(i)}$ sends $1$ to agents with indices $\le i+\lfloor \frac{n}{k}\rfloor$ and $0$ to others: $\bmuone^{(i)}(
    [i+\lfloor n/k\rfloor])=1$;
    \item For all $i\le j\le i+\lfloor \frac{n}{k}\rfloor$,
    $\bmuzero^{(i)}(
    [j])=\theta_j-\theta_{j+1}$. The remaining probability mass is assigned to the empty set: $\bmuzero^{(i)}(\emptyset)=1-(\theta_i-\theta_{i+\lfloor \frac{n}{k}\rfloor+1})$.
\end{itemize}
\end{restatable}

To gain some intuition for this approach, consider an agent $i$ who receives a positive signal from the optimal private scheme $(\bmuzero^\star,\bmuone^\star)$. Based on the structure of prefixes, this automatically implies that every agent with a smaller index $i'\le i$ also receives $s_{i'}=1$. Thus, the leaked signals from agents with smaller indices carry no information and do not affect agent $i$'s best response.
However, the optimal private scheme is brittle when an agent with a larger index leaks their signal, as they are likely to receive signal $0$ (when not included in the prefix drawn from $\bmuzero^\star$), which immediately fully reveals the state $\omega_0$ and causes all adopters to stop adopting.
Fortunately, the $k$ agents who leak signals are not chosen adversarially but drawn from $N$ uniformly at random, which implies that there is a constant probability that none of the $k$ agents fall within the range $[i,i+\frac{n}{k}]$.
Leveraging this observation, the sender can preserve the utility from adopter sets $[i],[i+1],\ldots,[i+\frac{n}{k}]$ if they are willing to sacrifice the utility from agents with indices $>i+\frac{n}{k}$ and deterministically send them signal $0$ in both states $\omega_0$ and $\omega_1$, effectively removing any information carried by their signals.

We formally prove \Cref{lemma:masking-removing-randomness} in \Cref{app:proof-masking-norandom}. The following theorem uses this lemma to upper bound the sender's price of robust downstream utility.
\begin{restatable}{theorem}{GapExpectedBroadcast}
\label[theorem]{thm:gap-expected-broadcast}
    For any $k\le n$ and any instance with a supermodular sender utility function, we have $\PoDR(\cG)\le O(k)$ for $\cG\in\{\kstar,\kclique,\kbroadcast\}$.
\end{restatable}
\begin{proof}[Proof of \Cref{thm:gap-expected-broadcast}]
    If $k>\frac{n}{2}$, we can invoke \Cref{thm:upper-bound-private-public} to construct a public scheme that achieves $\OPTexpected(\cG)\ge\OPTpublic\ge\Omega(1/n)\cdot\OPTprivate= \Omega(1/k)\cdot\OPTprivate$. It remains to prove the theorem for $k\le\frac{n}{2}$.
    As before,
    it suffices to construct a signaling scheme that achieves expected utility $\Omega(1/k)\cdot\OPTprivate(\omega_0)$.
    According to \citet{AB19}, we can explicitly express $\OPTprivate(\omega_0)$ as
        $\OPTprivate(\omega_0)=(1-\lambda)\sum_{i=1}^{n}(\theta_i-\theta_{i+1})\cdot V([i]).$
    Therefore, there exists some $i^\star\in N$ such that the sum from $i^\star$ to $i^\star+\lfloor \frac{n}{k}\rfloor$ already makes up an $\Omega(1/k)$ fraction of {the entire summation (from $1$ to $n$)}.
    The scheme $\left(\bmuzero^{(i^\star)},\bmuone^{(i^\star)}\right)$ constructed in \Cref{lemma:masking-removing-randomness} thus achieves a utility of $\Omega(1/k)\cdot\OPTprivate(\omega_0)$ and proves an $O(k)$ bound on the price of robust downstream utility.
\end{proof}

\paragraph{\textbf{Masking by matching randomness.}}
The second approach achieves an improved $O(1)$ $\PoDR$ in the $k$-star model, where the $k$ leaked signals are observable only to one randomly chosen agent (the center of the star).
This approach masks information in a more nuanced way: instead of using fully deterministic signals, we set $\bmuone$ to also assign probabilities to prefix subsets of adopters, with proportional probabilities under both $\bmuzero$ and $\bmuone$.
The goal is to ensure that, regardless of which prefix subset is realized, there is always a constant probability that all adopters are preserved---either the center is not an adopter, making its deviation irrelevant; or the center is an adopter who remains willing to adopt even after gaining more information about the realized prefix through $k$ random leakages.
The following lemma describes the concrete signaling scheme used to establish this claim:
\begin{restatable}[Masking by matching randomness]{lemma}{maskMatchingRandomness}
\label[lemma]{lemma:masking-matching-randomness}
Let $c_0,c_1$ be constants such that $0<c_0\le c_1<1$ and $c_0+c_1\le1$, and $m\in N$ be a given cutoff. Consider the following signaling scheme:
\begin{itemize}
    \item For each $i\in N$, $\bmuzero([i])=c_0(\theta_i-\theta_{i+1})$. For the empty set,  $\bmuzero(\emptyset)=1-c_0\cdot \theta_1$.
    \item For each $m\le i<n$, $\bmuone([i])=\frac{c_1}{\theta_m}(\theta_i-\theta_{i+1})$. For the full set $N=[n]$, $\bmuone(N)=1-c_1+c_1\cdot\frac{\theta_n}{\theta_m}$.
\end{itemize}
This scheme satisfies (1) For every agent $i\in N$, $a_i^{\bmu}(s_i=1)=1$, i.e., agents receiving a positive signal will adopt if they do not observe any leakages; (2) For all $i\le m$ and $j\ge m$, we have $\bmuzero([j])\le \theta_i\cdot \bmuone([j])$, i.e., agent $i$'s best response is to adopt even if $i$ fully observe the prefix $[j]$ that receive positive signals.
\end{restatable}

We end this section by briefly explaining why the scheme in \Cref{lemma:masking-matching-randomness} preserves each prefix with a constant probability. The formal proof is deferred to \Cref{app:proof-masking-match-random}.
Let $[j]$ be the realized prefix and random variable $i$ be the center of the star. The first property in \Cref{lemma:masking-matching-randomness} guarantees that $[j]$ is preserved if $i\not\in[j]$, which occurs with constant probability when $j$ is small. Therefore, it suffices to show that the set of adopters is also preserved when $j$ is large and $i\in[j]$.
Upon observing $k$ leaked signals, agent $i$ has an updated knowledge that $j\in[l,r]$ where $l$ is the largest index of agents who leak a positive signal, and $r+1$ is the smallest index of agent who leak a negative signal. Therefore, agent $i$ would best respond with action $1$ if $\sum_{j\in[l,r]}\bmuzero([j])\le \theta_i\cdot\sum_{j\in[l,r]}\bmuone([j])$. If $l\ge m$, the second property of \Cref{lemma:masking-matching-randomness} guarantees that the inequality holds for each term, so agent $i$ will adopt. It remains to bound the probability that $l\ge m$, which holds when at least one agent in $[m,j]$ leak a positive signal. This is achieved by optimizing the cutoffs for ``large'' $j$ and $m$ in the signaling scheme. We use \Cref{lemma:masking-matching-randomness} to prove the following theorem in \Cref{app:proof-gap-expected-star}.
\begin{restatable}{theorem}{gapExpectedStar}
\label[theorem]{thm:gap-expected-star}
    For any $k$ and any instance where the sender's utility function is supermodular, we have $\PoDR(\kstar)=O(1)$.
\end{restatable}

\section{Lower Bounds on the Price of Robustness}
\label{sec:technical-lower-bound}
We begin with a hard instance with a supermodular sender utility, which we will use in \Cref{sec:lower-bound-persuasive,sec:lower-bound-expected} to establish several lower bounds on the price of robustness.
We will present the lower bounds for submodular sender utilities in \Cref{app:lower-bound-submodular}.

\begin{example}[Hard instance with supermodular sender utility]
    \label[example]{ex:hard-instance-supermodular}
    In this instance, each receiver $i \in N$ has persuasion level $\theta_i=2^{-i}$, and the sender has the following supermodular utility function:
\begin{align}
    V(S):= \sum_{i \in N:[i]\subseteq S} 2^i
    =\sum_{i \in N:[i]\subseteq S}\theta_i^{-1}.
    \label{eq:util-supermodular}
\end{align}
We also set the prior probability $\lambda=\pr{\prior}{\omega_1}:=2^{-n}$ to be very small, so that most of the sender's utility must be gained from adopters under state $\omega_0$. In this instance, the optimal private solution from \cite{AB19} gives $\OPTprivate\asymp \sum_{i\in N}(\theta_i-\theta_{i+1})\cdot V([i])=\Theta(n)$.
\end{example}

The instance constructed in \Cref{ex:hard-instance-supermodular} has two key properties that will be extremely useful:
\begin{itemize}
    \item[(1)] The ``marginal contribution'' from persuading an agent $i$ is inversely proportional to the persuasion level $\theta_i$.\footnote{Assuming all agents prior to $i$ are also persuaded to adopt.} This guarantees that the optimal private scheme (which fully utilizes the persuasion level $\theta_i$) can extract constant utility simultaneously from each agent, implying $\OPTprivate=\Theta(n)$. This property will also be useful in upper bounding the sender's utility in the presence of leakages, as we will see later.
    \item[(2)] The persuasion levels decrease exponentially with the agent's index. Combining with property (1), this implies that the utility gained from an adopter subset $S$ is dominated by the contribution from the agent $i\in S$ with the largest index for which $[i]\subseteq S$. Specifically, we have $V(S)\le 2^{i+1}$.
\end{itemize}

\subsection{Lower Bounds on the Price of Robust Persuasiveness}
\label{sec:lower-bound-persuasive}
To see why the two properties of \Cref{ex:hard-instance-supermodular} are useful, we first consider the extreme case of public persuasion where all signals are leaked. The following lemma shows the lower bound on the price of robust persuasiveness ($\PoWR_k$) at $k = n - 1$.
\begin{theorem}
\label[theorem]{thm:lower-bound-private-public}
    In \Cref{ex:hard-instance-supermodular}, we have $\PoWR_{n-1} \coloneqq \OPTprivate/\OPTpublic\ge\Omega(n)$.
\end{theorem}

\begin{proof}[Proof of \Cref{thm:lower-bound-private-public}]
Since the prior probability of state $\omega_1$ is $\lambda = 2^{-n}$, the sender gains a utility of at most $\lambda\cdot V(N) = O(1)$ under state $\omega_1$. Therefore, we will focus on the utility under $\omega_0$.
    Let $(\bmuzero,\bmuone)$ be any publicly persuasive scheme and let $S$ be a realized subset of adopters.
To ensure that all agents in $S$ follow the signal, the probability of $S$ is constrained by the adopter with smallest persuasion level: $\bmuzero(S)\le \min_{i\in S}\theta_i \cdot\bmuone(S)=\theta_{\max(S)}\cdot\bmuone(S)$. Additionally, by property~(2) of \Cref{ex:hard-instance-supermodular}, the sender's utility satisfies $V(S)\le 2^{\max (S)+1}$. These two observations together yield an upper bound on sender's utility from
persuading each subset $S$:
\begin{align}
    \bmuzero(S) \cdot V(S)\le \theta_{\max(S)}\cdot\bmuone(S)\cdot 2^{\max (S)+1}=2\cdot \bmuone(S),
    \label{eq:util-subset-bound}
\end{align}
where the last step applies property~(1), which ensures that the marginal utility $2^{\max(S)}$ is the inverse of the persuasion level $\theta_{\max(S)}=2^{-\max(S)}$.

Summed over all subsets $S\in 2^N$, \cref{eq:util-subset-bound} implies a constant upper bound on the sender's total utility and thus an $\Omega(n)$ gap between $\OPTprivate$ and $\OPTpublic$:
\[
\sum_{S\in 2^N}\bmuzero(S) \cdot V(S)
\le2 \sum_{S\in 2^N}\bmuone(S)=2.\qedhere
\]
\end{proof}

Now consider $k$-worst-case persuasive schemes. We present the lower bound in \Cref{thm:lower-bound-persuasive-supermodular}.

\begin{restatable}{theorem}{lowerBoundPersuasiveSupermodular}
\label[theorem]{thm:lower-bound-persuasive-supermodular}
     For all $k$, the instance in \Cref{ex:hard-instance-supermodular} gives $\PoWR_k\ge\Omega(\min\{n,2^k\})$.
\end{restatable}

\begin{proof}[Proof of \Cref{thm:lower-bound-persuasive-supermodular} when $n=2^k$]
We consider the special case where $k = \log_2 n$ (assuming that $n = 2^k$ is a power of $2$); the general case is handled in \Cref{app:proof-lower-bound-persuasive-supermodular}. Since $k$ leakages do not fully reveal the realized adopter subset $S$, we no longer have \cref{eq:util-subset-bound} for each individual $S$. Instead, we will partition the family of all adopter subsets $2^N$ (equivalently, the hypercube $\{0,1\}^n$) into $n$ \emph{subcubes} $C_1,\ldots,C_n$ and establish a similar inequality to \cref{eq:util-subset-bound} for each subcube.

More concretely, we will construct a subcube partition (see \Cref{lemma:existence-of-subcube-partition} for the formal statememt) in which each subcube $C_i$ corresponds to the knowledge of receiver $i$ after observing $k$ leakages: there exists $I_i$ with size $|I_i|=k$, such that each subcube $C_i$ contains all signal realizations consistent with agent $i$ observing $s_i=1$ and the additional leakages $I_i$. In addition, each $I_i$ contains the leaked signal $s_{i+1}=0$.

Since $C_i$ represents agent $i$'s knowledge set after receiving a positive signal and observing $I_i$, the worst-case persuasiveness condition requires $\sum_{S\in C_i}\bmuzero(S)\le\theta_i\cdot\sum_{S\in C_i}\bmuone(S)$.
As for the utility $V(S)$ for each $S\in C_i$, since $s_{i+1}=0$, we have $V(S)\le 2^{i+1}$. Together, they imply
\[
\sum_{S\in C_i}\bmuzero(S) \cdot V(S)\le 2^{i+1}\sum_{S\in C_i}\bmuzero(S) \le2^{i+1}\cdot\theta_i\cdot\sum_{S\in C_i}\bmuone(S) = 2\sum_{S\in C_i}\bmuone(S),
\]
where the last step again uses property (1) that $\theta_i=2^{-i}$ is the inverse of marginal utility $2^i$. Thus, we have established a similar inequality per subcube. Summing over all subcubes shows $\OPTpersuasive_k\le O(1)$ and $\PoWR_k=\OPTprivate/\OPTpersuasive_k=\Omega(n)=\Omega(2^k)$.
\end{proof}

\subsection{Lower Bounds on the Price of Robust Downstream Utility}
\label{sec:lower-bound-expected}
In this section, we discuss how to adapt the hard instance in \Cref{ex:hard-instance-supermodular} to establish lower bounds for the $k$-broadcast and $k$-clique random leakages models.
The high-level idea remains to bound the sender's utility using a variant of inequality~\eqref{eq:util-subset-bound}. However, with random (rather than worst-case) leakages, we can no longer handpick specific leakage patterns to create a subcube partition. Instead, the random leakage patterns naturally induce a randomized partition.

In both models, our lower bounds apply to \emph{prefix-based} schemes, which we define below.
\begin{definition}[Prefix-based schemes]
    \label{def:prefix-based}
    A signaling scheme $(\bmuzero,\bmuone)$ is prefix-based if both $\bmuzero$ and $\bmuone$ are distributions over prefixes (including the empty set), i.e., $\bmuzero,\bmuone\in\Delta\left(\{\emptyset\} \cup \{[i]\mid i\in N\}\right)$.
\end{definition}

In the following theorem, we present the lower bound for prefix-based schemes in the $k$-broadcast model, and provide a proof illustrating how to establish lower bounds using the implicit partition induced by random leakages.
\begin{restatable}{theorem}{lowerBoundBroadcast}
\label[theorem]{thm:lower-bound-broadcast}
    In the instance described in \Cref{ex:hard-instance-supermodular}, for any $k\in [n]$, all prefix-based schemes (see \Cref{def:prefix-based}) must incur $\PoDR(\kbroadcast)\ge\Omega(k)$.
\end{restatable}

\begin{proof}[Proof of \Cref{thm:lower-bound-broadcast}]
    We start by upper bounding the sender's expected utility conditioned on the identities of the $k$ agents (denoted with $v_1 < v_2 < \cdots < v_k$) who leak their signals in $k$-broadcast model. We will take an expectation over the randomness in $\{v_1, v_2, \ldots, v_k\}$ at the end.

    Suppose that the realized prefix from the prefix-based scheme is $[j]$. The leaked signals from $\{v_1,\ldots,v_k\}$ effectively partition $N$ into $k+1$ blocks of the form $B_l=[v_l,v_{l+1})$, {where $l \in \{0, 1, \ldots, k\}$, $v_0 = 0$ and $v_{k+1} = n + 1$}. From the perspective of a receiver, the leaked signal values $s_{v_1},\ldots,s_{v_k}$ narrow the possibility of the prefix length $j$ down to a block $B_{l}=[v_l,v_{l+1})$, where $s_{v_l}=1$ but $s_{v_{l+1}}=0$ {(we take $s_{v_0} = s_0 = 1$ and $s_{v_{k+1}} = s_{n+1} = 0$)}.

    If an agent $i$ happens to also lie in $B_l$, their own signal further narrows the possible prefixes down to either $j\in[v_j,i-1]$ or $j\in[i,v_{j+1})$, depending on whether $s_i=0$ or $s_i=1$. Let $C_i\subseteq\{1,\ldots,n\}$ represent all possible values of $j$ that are consistent with agent $i$'s observations. Then, we have $\{j\}\subseteq C_i\subseteq B_l$. Therefore, for an agent $i$ to adopt, a necessary condition is \begin{align}
    \bmuzero([j])\le\sum_{j'\in C_i} \bmuzero([j'])\le
    \theta_i\cdot\sum_{j'\in C_i} \bmuone([j'])
    \le \theta_i\cdot\sum_{j'\in B_l} \bmuone([j']).
    \label{eq:cond-follow-random}
\end{align}

Let $i^\star$ be the largest index that satisfies \Cref{eq:cond-follow-random}. By property~(2) of \Cref{ex:hard-instance-supermodular}, the sender's utility when prefix $[j]$ is chosen is upper bounded by $2^{i^\star+1}$. The contribution to the sender's utility is thus bounded by
\[
\bmuzero([j])\cdot 2^{i^\star+1}\le 2^{i^\star+1}\cdot\theta_{i^\star}\cdot\sum_{j'\in B_l} \bmuone([j'])=2\sum_{j'\in B_l} \bmuone([j']).
\]
Summing over $j$ and exchanging the order of summation, we obtain
\begin{align*}
    2\sum_{j\in N}\sum_{j'\in B_{l(j)}} \bmuone([j'])=
    2\sum_{j'=0}^n \bmuone([j'])\sum_{j\in N}\1{j'\in B_{l(j)}}
    =2\sum_{j'=0}^n \bmuone([j'])\cdot|B_{l(j')}|,
\end{align*}
where $B_{l(j)}$ makes the dependency on $j$ explicit and represents the block that contains $j$.

In the last step, we take the expectation over the randomness in $v_1, \ldots, v_k$, which determine the blocks $B_0, \ldots, B_k$. By the linearity of expectation, the sender's expected utility is upper bounded by
\begin{align*}
2\sum_{j'=0}^n \bmuone([j'])\cdot\Ex{v_1,v_2,\ldots,v_k}{|B_{l(j')}|}
=2\sum_{j'=0}^n \bmuone([j'])\cdot O\left(\frac{n}{k}\right)=O\left(\frac{n}{k}\right),
\end{align*}
where $\Ex{}{|B_{l(j')}|}=O(n/k)$ represents the expected length of the block containing $j'$. Thus, we have proved that every prefix-based scheme has an expected utility of $O(n / k)$, which, combined with $\OPTprivate=\Theta(n)$, establishes a lower bound of  $\PoDR(\kbroadcast)\ge\Omega(k)$.
\end{proof}

Next, we move on to the $k$-clique model. Unlike the $k$-broadcast setting, here only agents within the clique observe the leaked signals, and only their actions can be characterized by \cref{eq:cond-follow-random}. This means that if $i^\star$ continues to be the largest index that satisfies \cref{eq:cond-follow-random}, the actual largest index of an adopter can exceed $i^\star$, potentially leading to an exponential increase in the sender's utility.

To address this issue, we modify the instance in \Cref{ex:hard-instance-expected} to control the rate at which the sender's utility increases with the length of the prefix.
The modified instance is inspired by the following thought experiment: We partition the $N$ agents into $k$ blocks and restrict the sender to send identical signals to agents within each block. Then, each block of agents can be viewed as a ``super-agent''. Under this setup, the $k$-clique model effectively reduces to public persuasion among $k$ ``super-agents'', if there is exactly one leaked signal in every block. By \Cref{thm:lower-bound-private-public}, this setup implies an $\Omega(k)$ lower bound on the price of robust downstream utility. Following this intuition, we construct a hard instance by evenly partitioning $N$ into $k$ ``super-agents'' of size $\Theta(n/k)$, which reflects the ``typical'' leakage patterns, and embedding the previous hard instance in \Cref{ex:hard-instance-supermodular} with $k$ agents.

\begin{restatable}[Hard instance for the $k$-clique model]{example}{hardInstanceKClique}
\label[example]{ex:hard-instance-expected}
    We divide the agents into contiguous blocks of size $B = 4\lceil\frac{n}{k}\rceil$: For each $j\ge0$, let the $j$-th block contains agents indexed in $[j\cdot B,\ \min\{(j+1)\cdot B,n+1\})$.
    Agents in each block of agents share the same persuasion level. Specifically, each agent $i\in N$ lies in block $b_i=\lfloor\frac{i}{B}\rfloor$
    and thus has persuasion level $\theta_i:=2^{-b_i}=2^{-\lfloor\frac{i}{B}\rfloor}$. The sender's utility function is defined in terms of the maximum number of blocks contained in $S$:
       $ V(S):=\sum_{b=1}^k \1{[b\cdot B]\subseteq S}\cdot 2^b.$
\end{restatable}

 We are now ready to present the lower bound below. Its proof is deferred to \Cref{app:proof-lower-bound-clique}.

\begin{restatable}{theorem}{lowerBoundClique}
\label[theorem]{thm:lower-bound-clique}
    For any $k$, the instance constructed in \Cref{ex:hard-instance-expected} guarantees that all prefix-based schemes (see \Cref{def:prefix-based}) must suffer from $\PoDR(\kclique)\ge\Omega\left(k / \log k\right)$.
\end{restatable}

\section{Discussion and Open Problems}\label{sec:discussion}
Recall from \Cref{table:results} that, while we obtained nearly-tight bounds on the price of robust persuasiveness ($\PoWR_k$), we still miss a few lower bounds on $\PoDR$ for the expected downstream utility robustness formulation. Here, we discuss why the latter formulation appears inherently harder and mention several concrete open problems that necessitate a deeper structural understanding of signaling schemes that maximize the expected downstream utility.

The optimization problem for expected downstream utility robustness is \emph{highly discontinuous}, and a straightforward algorithm for finding the optimal solution would involve solving \emph{doubly exponentially} many LPs. This complexity also leads to at least two counter-intuitive aspects of the optimal solution, both witnessed by small and simple instances:
\begin{itemize}
[noitemsep,topsep=0pt,leftmargin = *]
    \item \textbf{There is \emph{no} ``revelation principle'':} Indirect schemes---in which the signal space is larger than the actions space---can be strictly better than direct ones.
    \item \textbf{The optimal utility is \emph{not} monotone in the leakage pattern:} The sender can achieve a strictly higher utility under more leakages, as previously shown by~\citet{KT23}.
\end{itemize}

\paragraph{The computational aspect.} As shown in~\eqref{eq:lp-persuasive}, $\OPTpersuasive_k$ can be found by solving an LP with $O(2^n)$ variables and $n^{O(k)}$ constraints.
In contrast, given a distribution $\cG$ over leakage patterns, finding $\OPTexpected(\cG)$ appears to be much harder. For simplicity, we assume a size-$2$ signal space, so that $(\bmuzero, \bmuone)$ is still $O(2^n)$-dimensional. However, the expected downstream utility achieved by $(\bmuzero, \bmuone)$ is no longer linear: it depends on the receivers' best responses to the signals, which, by \Cref{eq:br-leakage}, is the indicator of a linear inequality in $(\bmuzero, \bmuone)$. This makes the optimization problem extremely discontinuous. Note that the optimization does become linear once we enumerate the best responses of the receivers under every realization of the signals and leakage patterns, and enforce these responses via linear constraints. However, there are $(2^n)^{2^n} = 2^{\Omega(2^n)}$ different possible combinations. Improving this runtime to $2^{O(n)}$ would require a significant reduction of the search space, e.g., via a better understanding of the optimization landscape.

\paragraph{Direct vs indirect schemes.} This work focuses on \emph{direct} signaling schemes.
When leakages are absent (i.e., in the private Bayesian persuasion setup of~\cite{AB19}), restriction to direct signaling schemes is without loss of generality due to the revelation principle. In \Cref{sec:direct-vs-indirect}, we provide evidence suggesting that there is \emph{no} analogue of the revelation principle when the expected downstream utility robustness is concerned.
At a high level, signal leakages enable the sender to recommend actions to each receiver~$i$ not only via a single signal $s_i$, but the combination of $s_i$ and all other signals observable by receiver~$i$. For instance, when the leakage $(j \to i)$ is present, the sender may recommend receiver~$i$ to take action $1$ if $s_i = s_j$ and take action $0$ if the signals differ. To achieve an optimal expected downstream utility, the sender may need to encode such recommendations optimally via the pairwise (or even higher-order) relations between the signals, and this would require a larger alphabet for the signals. Many natural questions remain unanswered regarding the trade-off between the lack of directness (formalized by the size of the signal space) and the optimality of sender's utility: Can the sender achieve a constant approximation of the optimal utility via a two-signal scheme? Is there a function $f(n)$ such that a size-$f(n)$ signal space is sufficient for achieving the exact optimal utility?

\paragraph{Monotonicity of sender utility.} Do signal leakages always make it harder to persuade the receivers? Perhaps counter-intuitively, this is not the case. As shown by~\citet{KT23}, there exist two instances on which the optimal sender utility is actually higher when strictly more leakages happen.\footnote{In hindsight, this should be unsurprising---if one views Bayesian persuasion as information getting ``leaked'' from the sender to the receivers, every non-trivial signaling scheme witnesses that a higher utility can be achieved compared to the case where ``leakages'' (i.e., persuasion) do not happen.} In \Cref{sec:somewhat-indirect}, we give an even smaller instance (with $n = 3$ receivers) that witnesses this phenomenon.
As a consequence of this non-monotonicity, in \Cref{table:results}, our results for one leakage distribution do not immediately imply results for seemingly ``easier'' or ``harder'' setups.\footnote{For example, we needed two separate proofs for the $\tilde\Omega(k)$ lower bounds in $\kclique$ and $\kbroadcast$, as the former may not directly imply the latter.} Does some form of monotonicity hold for the natural and structured classes that we consider? Does leakage-robust persuasion always become harder (i.e., $\OPTexpected(\cG)$ decreases) as the parameter $k$ grows? Is $\kbroadcast$ always harder than $\kclique$, and is $\kclique$ always harder than $\kstar$?

\section*{Acknowledgments}
This work was supported in part by the National Science Foundation under grant CCF-2145898, by the Office of Naval Research under grant N00014-24-1-2159, a C3.AI Digital Transformation Institute grant, and Alfred P. Sloan fellowship, and a Schmidt Science AI2050 fellowship.

\bibliographystyle{abbrvnat}
\bibliography{arxiv-ref}
\newpage
\appendix

\section{Definition of $k$-Worst-Case Persuasiveness}
\label{app:def-k-persuasive}
In this section, we provide further justifications for the definition of $k$-worst-case persuasiveness in \Cref{def:k-worst-case-persuasiveness}.

\begin{proposition}
\label[proposition]{prop:maxmin-utility}
   The benchmark $\OPTpersuasive_k$ is equal to the sender's maxmin effective utility under the worst-case leakage pattern with in-degree at most $k$. Formally, if we define
   \begin{align*}
       \maxminone_k&:=\sup_{\bmu \text{ $k$-worst-case persuasive}}
\min_{G\atop\text{in-deg}(G)\le k}\left\{
\Ex{\omega\sim\prior\atop (s_1,\ldots,s_n)\sim\bmuomega}{V(a_1^{\bmu}(s_1,I_1^G), a_2^{\bmu}(s_2,I_2^G), \ldots, a_n^{\bmu}(s_n,I_n^G))}
\right\}\\
\maxmintwo&:=\sup_{\bmu\text{ $k$-worst-case persuasive}}
\left\{
\Ex{\omega\sim\prior\atop (s_1,\ldots,s_n)\sim\bmuomega}{\min_{G\atop\text{in-deg}(G)\le k}
V(a_1^{\bmu}(s_1,I_1^G), a_i^{\bmu}(s_2,I_2^G),\ldots,a_n^{\bmu}(s_n,I_n^G))}
\right\},
   \end{align*}
   we have $\OPTpersuasive_k=\maxminone_k=\maxmintwo_k$.
\end{proposition}
\begin{proof}[Proof of \Cref{prop:maxmin-utility}]
    We first show that $\maxminone_k\ge \OPTpersuasive_k$ and $\maxmintwo_k\ge \OPTpersuasive_k$. For any leakage pattern $G$ with in-degree at most $k$, we have $|I_i^G|\le k$ for all $i\in N$. By the second requirement of
    $k$-worst-case persuasiveness (see \Cref{def:k-worst-case-persuasiveness}), this implies $a_i^{\bmu}(s_i,I_i^G)\ge a_i^{\bmu}(s_i)$. Since the sender utility $V$ is monotone, we have $V(a_1^{\bmu}(s_1,I_1^G), a_2^{\bmu}(s_2,I_2^G), \ldots, a_n^{\bmu}(s_n,I_n^G))\ge V(s_1,\ldots,s_n)$. Taking the maximum for all $k$-worst-case persuasive $\bmu$ establishes $\maxminone_k\ge\OPTpersuasive_k$ and $\maxmintwo_k\ge \OPTpersuasive_k$.

    It remains to show that equality holds for some leakage $G$ with in-degree at most $k$. This is achieved when $G$ is the empty graph with no edges, which corresponds to the fully private setting. Condition (1) of $k$-worst-case persuasiveness then guarantees $a_i^{\bmu}(s_i,I_i^G)=a_i^{\bmu}(s_i)=s_i$. This establishes the equality.
\end{proof}

\begin{proposition}[Two-sided version of $k$-worst-case persuasiveness reduces to public]
\label[proposition]{prop:two-sided-persuasiveness}
    Suppose the persuasion levels $\theta_1$ through $\theta_n$ are all distinct, and $k\ge2$. Let  $\bmu=(\bmuzero,\bmuone)$ be any signaling scheme satisfies $a_i^{\bmu}(s_i,I_i)=s_i$ for all $i\in N, s_i\in\{0,1\},$ and $|I_i|\le k$. Then ${\bmu}$ must be publicly persuasive.
\end{proposition}
\begin{proof}[Proof of \Cref{prop:two-sided-persuasiveness}]
    We first show that both $\bmuzero$ and $\bmuone$ must be supported on prefixes, i.e., {$\support(\bmuzero),\support(\bmuone)\subseteq \{\emptyset\} \cup \{[i]\mid i\in N\}$}. For the sake of contradiction, suppose that there exist $i<j$ and $S\in \support(\bmuzero)\cup\support(\bmuone)$ such that $i\not\in S, j\in S$, implying that $S$ is not a prefix.
    When such an $S$ is realized from $\bmu$, the signals received by $i$ and $j$ would be $s_i=0$ and $s_j=1$.
    Consider a leakage pattern where $I_i=\{(j,s_j=1)\}$ and $I_j=\{(i,s_i=1)\}$.

    Now consider the implications of two-sided worst-case persuasiveness.
    On the one hand, since $a_i(s_i,I_i)=s_i=0$, we should have
    \begin{align}
        \sum_{s_{-(i,j)}}\bmuzero(s_i=0,s_j=1,s_{-(i,j)})
        \ge\theta_i\cdot\sum_{s_{-(i,j)}}{\bmuone}(s_i=0,s_j=1,s_{-(i,j)}).
        \label{eq:tmp-one-hand}
    \end{align}
    On the other hand, since $a_j(s_j,I_j)=s_j=1$, we have
    \begin{align}
        \sum_{s_{-(i,j)}}\bmuzero(s_i=0,s_j=1,s_{-(i,j)})
        \le\theta_j\cdot\sum_{s_{-(i,j)}}{\bmuone}(s_i=0,s_j=1,s_{-(i,j)}).
        \label{eq:tmp-other-hand}
    \end{align}
    Since $i<j$ and the persuasion levels are distinct, we have $\theta_i>\theta_j$. However, since $S\in\support(\bmuzero)\cup\support(\bmuone)$, at least one side of the inequalities \eqref{eq:tmp-one-hand} and \eqref{eq:tmp-other-hand} must be nonzero, causing a contradiction. Therefore, we conclude that both $\bmuzero$ and $\bmuone$ have to be supported on prefixes.

    Now for the prefix-based scheme $\bmu$, consider any non-empty prefix $S=[l]$ with $l \in [n]$. For each $i\in S$, the leakage $I_i=\{(l,s_l=1),(l+1,s_{l+1}=0)\}$ (or $I_i=\{(n,s_n=1)\}$ if $l=n$) uniquely identifies $l$ and thus $S$. It then follows from the requirement of $k$-worst-case persuasiveness that $\bmuzero(S)\le\theta_i\cdot\bmuone(S)$ for all $i\in S$ --- i.e., $i$ will still follow the signal even after they observe the full signal realization. As a result, $\bmu$ is publicly persuasive.
\end{proof}

\section{Details for Section~\ref{sec:technical-subsampling}}
\subsection{Property of Submodular and XOS Functions}
\label{app:submodular-subadditive-property}
\begin{definition}[Submodularity]
    Let $V:2^N\to\mathbb{R}_+$ be a nonnegative set function. We say that $V$ is submodular if for every two subsets $S,T\subseteq N$, it holds that
    \[
    V(S)+V(T)\ge V(S\cup T)+V(S\cap T).
    \]
\end{definition}

\begin{definition}[XOS]
    Let $V:2^N\to\mathbb{R}_+$ be a nonnegative set function. We say that $V$ is XOS (or fractionally subadditive) if there exists a collection of $K$ additive set functions $\{V^k(S)=\sum_{j\in S} v_j^k\mid k\in[K]\}$ such that for all $S\subseteq N$,
    \[
    V(S)=\max_{k\in[K]} V^k(S).
    \]
\end{definition}

\begin{lemma}
\label[lemma]{lemma:property-xos}
    Let $V:2^N\to\mathbb{R}_+$ be a nonnegative set function that is submodular or XOS, and let $S\subseteq N$ be any subset. Let $\D_\gamma(S)$ be the distribution of subset $S'$ that is generated by including each element in $S$ independently with probability $\gamma$. Then
    \begin{align*}
        \Ex{S'\sim\D_\gamma(S)}{V(S')}\ge\gamma\cdot V(S).
    \end{align*}
\end{lemma}
\begin{proof}
    Since submodular functions are a subclass of XOS functions \cite{nisan2000bidding,lehmann2001combinatorial}, it suffices to prove the claim for XOS functions. For each $k\in[K]$, the linearity of expectation gives us
    \begin{align}
        \Ex{S'\sim\D_\gamma(S)}{V^k(S')}=\Ex{S'\sim\D_\gamma(S)}{\sum_{j\in S}\1{j\in S'}v_j^k}=
    \sum_{j\in S}\gamma\cdot v_j^k=\gamma\cdot V^k(S).
    \label{eq:tmp-linear-exp}
    \end{align}
    Therefore, by Jensen's inequality and the definition of XOS functions,
    \begin{align*}
        \Ex{S'\sim\D_\gamma(S)}{V(S')}
        =\Ex{}{\max_{k\in[K]} V^k(S')}
        \ge \max_{k\in[K]}\Ex{}{ V^k(S')}
        =\gamma\cdot\max_{k\in[K]}V^k(S)
        =\gamma\cdot V(S),
    \end{align*}
    where the second to last step uses \Cref{eq:tmp-linear-exp}. This completes the proof for XOS functions and, therefore, for the special case of submodular functions as well.
\end{proof}

\subsection{Proof of \Cref{lemma:subsample-1-2}}
\label{app:proof-subsample-1-2}
\LemSubsampleHalf*

\begin{proof}
    We prove this lemma by verifying that $(\bmuzero,\bmuone)$ satisfies all the three conditions in the \eqref{eq:lp-persuasive} that characterizes $k$-worst-case persuasive schemes. Note that condition \circled{3} is clearly satisfied since $\bmuzero$ and $\bmuone$ are valid probability distributions. For condition \circled{2}, we have that $\forall i\in N$ and $|I_i|\le k$,
    \begin{align*}
        \sum_{s_{-i}:I_i\triangleright s_{-i}}\bmuzero(1,s_{-i})&\le
        \sum_{s_{-i}\in\{0,1\}^{n-1}} \bmuzero(1,s_{-i}) \tag{dropping condition $I_i\triangleright s_{-i}$} \\
        &= 2^{-(k+1)}\sum_{s_{-i}\in\{0,1\}^{n-1}} \bmuzero^\star(1,s_{-i})
        \tag{definition of $\bmuzero$}\\
        &\le 2^{-(k+1)}\cdot\theta_i \tag{$(\bmuzero^\star,\bmuone^\star)$ is privately persuasive}\\
        &\le  \theta_i\cdot 2^{-(|I_i|+1)}&\tag{$|I_i|\le k$}\\
        &= \theta_i\cdot\sum_{s_{-i}:I_i\triangleright s_{-i}}\bmuone(1,s_{-i}).
        \tag{definition of $\bmuone$}
    \end{align*}

    It remains to verify condition \circled{1} in \eqref{eq:lp-persuasive}. For all $i\in N$, when agent~$i$ receives a signal of value $1$, we have
    \begin{align*}
        \sum_{s_{-i}}\bmuzero(1,s_{-i})
        &= 2^{-(k+1)}\sum_{s_{-i}} \bmuzero^\star(1,s_{-i})
        \tag{definition of $\bmuzero$}\\
        &\le 2^{-(k+1)}\cdot\theta_i \tag{$(\bmuzero^\star,\bmuone^\star)$ is privately persuasive}\\
        &\le  \theta_i\cdot 2^{-1}\\
        &= \theta_i\cdot\sum_{s_{-i}}\bmuone(1,s_{-i}).
        \tag{definition of $\bmuone$}
    \end{align*}
    On the other hand, since $\bmuzero$ scales the probability of all nonzero signal realizations by a factor of $2^{-(k+1)}$, when agent $i$ receives $s_i=0$, we have
    \[
        \sum_{s_{-i}}\bmuzero(0,s_{-i})
    \ge \bmuzero(0,\ldots,0)
    =   1-2^{-(k+1)}\cdot[1-\bmuzero^\star(0,\ldots,0)]
    \ge 1-2^{-(k+1)}.
    \]
    It follows that
    \begin{align*}
        \sum_{s_{-i}}\bmuzero(0,s_{-i})
    &\ge 1-2^{-(k+1)}
    >\theta_i\cdot 2^{-1}\tag{$\theta_i\le1$}\\
    &=  \theta_i\cdot\sum_{s_{-i}}\bmuone(0,s_{-i}),\tag{definition of $\bmuone$}
    \end{align*}
    which establishes condition \circled{1}. The proof is now complete.
\end{proof}

\subsection{Proof of \Cref{lemma:subsample-1-k}}
\label{app:proof-subsample-1-k}
\LemSubsampleK*

\begin{proof}
    We prove this lemma by showing that $(\bmuzero,\bmuone)$ satisfies conditions $\circled{1}$ and $\circled{2}$ in \eqref{eq:lp-persuasive}; condition \circled{3} is clearly satisfied since, by construction, $\bmuzero$ and $\bmuone$ are valid probability distributions.

    We first verify condition \circled{2}. For all $i\in N$ and $I_i$, let $I_i^+$ denote the subset of signals in $I_i$ that have value $1$ and $I_i^-=I_i\setminus I_i^+$ be the set of signals with value $0$.
    We also slightly abuse the notation and let $\mu(S)$ (where $S\subseteq N$) denote the probability that $\mu$ assigns to the signal realization $(s_1,\ldots,s_n)$ where $s_i=\1{i\in S}$.

    According to the definition of $\bmuzero$, for an agent $i$ to observe $s_i=1$ and $I_i$, we need to first sample a set $S$ from ${(1-\gamma)^k}\cdot\bmuzero^\star(\cdot)$ such that $(I_i^+\cup \{i\})\subseteq S$, and then make sure that every element in $I_i^+\cup \{i\}$ remains after subsampling, while the elements in $S\cap I_i^-$ do not. This gives us
    \begin{align*}
        \sum_{s_{-i}:I_i\triangleright s_{-i}}\bmuzero(1,s_{-i})
    &=  \sum_{S\subseteq N:(I_i^+\cup \{i\})\subseteq S}(1-\gamma)^k\bmuzero^\star(S)\cdot \gamma^{1+|I_i^+|}\cdot(1-\gamma)^{|S\cap I_i^-|}\\
    &\le \sum_{S\subseteq N: i \in S}(1-\gamma)^k\bmuzero^\star(S)\cdot\gamma^{1+|I_i^+|} \tag{relax $(I_i^+\cup \{i\})\subseteq S$ to $i \in S$}\\
    &\le \sum_{S\subseteq N: i \in S}\bmuzero^\star(S)\cdot\underbrace{\gamma^{1+|I_i^+|}\cdot(1-\gamma)^{|I_i^-|}}_{=\sum_{s_{-i}:I_i\triangleright s_{-i}}\bmuone(1,s_{-i})}
        \tag{$|I_i^-|\le |I_i| \le k$}\\
    &=  \left(\sum_{s_{-i}}\bmuzero^\star(1,s_{-i})\right)\cdot\left(\sum_{s_{-i}:I_i\triangleright s_{-i}}\bmuone(1,s_{-i})\right)\\
    &\le  \left(\theta_i\cdot\sum_{s_{-i}}\bmuone^\star(1,s_{-i})\right)\cdot\left(\sum_{s_{-i}:I_i\triangleright s_{-i}}\bmuone(1,s_{-i})\right)
        \tag{$(\bmuzero^\star,\bmuone^\star)$ is privately persuasive}\\
    &\le  \theta_i\cdot\sum_{s_{-i}:I_i\triangleright s_{-i}}\bmuone(1,s_{-i}).
    \end{align*}
    We have thus established condition \circled{2}. Now we examine condition \circled{1}. The case where $s_i=1$ follows from \circled{2} with $I_i=\emptyset$. It suffices to prove the case for $s_i=0$.
    By the construction of $\bmuzero$, we have $\bmuzero(0, 0, \ldots, 0) \ge 1 - (1 - \gamma)^k$, and it follows that
    \begin{align*}
        \sum_{s_{-i}}\bmuzero(0,s_{-i})=&
        1-\sum_{s_{-i}}\bmuzero(1,s_{-i})\\
        =&1-\sum_{S':i\in S'}(1-\gamma)^k\cdot\bmuzero^\star(S')\cdot\gamma\\
    \ge& 1-\gamma
    \ge(1-\gamma)\cdot\theta_i\\
    =&\theta_i\cdot\sum_{s_{-i}}\bmuone(0,s_{-i}).
    \end{align*}
    We have established \circled{1}.

\end{proof}

\subsection{Gap between Private and Public Persuasion}

\begin{theorem}
\label[theorem]{thm:upper-bound-private-public}
For any supermodular sender utility function, we have $\OPTprivate/\OPTpublic\le O(n)$. Since public schemes are $k$-worst-case persuasive for all values of $k$, this bound implies $\OPTprivate/\OPTpersuasive_k\le O(n)$ for all $k\le n$.
\end{theorem}
\begin{proof}
    As in the proof of \Cref{thm:gap-persuasive-general}, it suffices to construct a signaling scheme that achieves utility $\Omega(n^{-1})\cdot\OPTprivate(\omega_0)$. By~\cite[Theorem~1]{AB19}, the optimal private utility under state $\omega_0$ can be written as
    \begin{align*}
        \OPTprivate(\omega_0)=(1-\lambda)\cdot\sum_{k=1}^{n}V([k])\cdot(\theta_k - \theta_{k+1}).
    \end{align*}
    Therefore, there exists $i^\star \in [n]$ such that $(1-\lambda)\cdot V([i^\star])\cdot(\theta_{i^\star}-\theta_{i^\star+1})\ge n^{-1}\cdot\OPTprivate(\omega_0)$.

    Consider the signaling scheme $(\bmuzero,\bmuone)$ defined as follows:
    \begin{itemize}
        \item $\bmuzero([i^\star])=\theta_{i^\star}$. The remaining probability of $1-\theta_{i^\star}$ is assigned to the empty set: $\bmuzero(\emptyset)=1-\theta_{i^\star}$.
        \item $\bmuone$ concentrates all the probability mass on the prefix $i^\star$: $\bmuone([i^\star])=1$.
    \end{itemize}
    To see why the scheme is publicly persuasive, note that $[i^\star]$ is the only nonempty subset that receives a non-zero probability in either $\bmuzero$ or $\bmuone$. Furthermore, for each receiver $i \in [i^\star]$, we have
    \[
        \bmuzero([i^\star])
    =   \theta_{i^\star}
    \le \theta_i
    =   \theta_i\cdot\bmuone([i^\star]),
    \]
    which guarantees that receiver~$i$ would adopt upon receiving the public signal $[i^\star]$.
    The sender's utility under state $\omega_0$ is at least
    \[
    (1-\lambda)\cdot\theta_{i^\star}\cdot V([i^\star])\ge (1-\lambda)\cdot V([i^\star])\cdot(\theta_{i^\star}-\theta_{i^\star+1})\ge n^{-1}\cdot\OPTprivate(\omega_0),
    \]
    where the second step follows from our choice of $i^\star$.

    Therefore, the better scheme between $(\bmuzero,\bmuone)$ and the full-information revelation scheme (which is also publicly persuasive) can achieve utility $\ge\Omega(n^{-1})\cdot\OPTprivate$, and thus establishes the upper bound of $\OPTprivate/\OPTpublic\le O(n)$.
    \end{proof}

\section{Details for \Cref{sec:technical-expected-util-robust}}
\subsection{Proof of \Cref{lemma:masking-removing-randomness}}
\label{app:proof-masking-norandom}
\maskRemovingRandomness*
\begin{proof}
Let $V=\{v_1,\ldots,v_k\}$ be the subset of agents who leak their signals. We will first show that $V\cap \left[i,i+\lfloor\frac{n}{k}\rfloor\right]=\emptyset$ with probability $\Omega(1)$. Since $V$ is sampled uniformly at random from all subsets of size $k$, this probability is given by
\begin{align*}
    \frac{\binom{n-\lfloor n/k\rfloor-1}{k}}{\binom{n}{k}}
    =\prod_{l=0}^{k-1}\frac{n-\lfloor n/k\rfloor-1-l}{n-l}
    \ge\left(1-\frac{\lfloor n/k\rfloor+1}{n-k+1}\right)^{k}
    \ge\left(1-\frac{2n}{k(n-k)}\right)^k
    \ge e^{-\frac{4n}{k(n-k)}\cdot k}\ge e^{-8},
\end{align*}
where we have used the fact that $1-x\ge e^{-2x}$ for $x\le\frac{1}{2}$. Therefore, it suffices to prove that when $V\cap\left[i,i+\lfloor\frac{n}{k}\rfloor\right]=\emptyset$, the signaling scheme $(\bmuzero,\bmuone)$ achieves a utility of at least
\[
\sum_{j=i}^{i+\lfloor n/k\rfloor}(1-\lambda)\cdot(\theta_j-\theta_{j+1})\cdot V([j]).
\]

For each prefix $[j]$ where $i\le j\le i+\lfloor\frac{n}{k}\rfloor$, we will show that all agents in $[j]$ will adopt when prefix $[j]$ is realized from $\bmuzero^{(i)}$. For each agent $x\in[j]$, after observing $s_x=1$ and the leaked signals $I_i=\{(v,s_v)\mid v\in V\}$, the prefixes that remain consistent with these observations are $[x],[x+1],\ldots,\left[i+\lfloor\frac{n}{k}\rfloor\right]$. We have
\begin{align*}
    \sum_{S \subseteq N: x \in S}\bmuzero^{(i)}(S)
&=  \sum_{j=x}^{i+\lfloor\frac{n}{k}\rfloor}\bmuzero^{(i)}([j])\\
&=  \sum_{j=x}^{i+\lfloor\frac{n}{k}\rfloor}(\theta_j-\theta_{j+1})\cdot\1{i \le j} \tag{definition of $\bmuzero^{(i)}$}\\
&\le\sum_{j=x}^{i+\lfloor\frac{n}{k}\rfloor}(\theta_j-\theta_{j+1})
=   \theta_x-\theta_{i+\lfloor\frac{n}{k}\rfloor+1}\\
&\le\theta_x=\theta_x\cdot\sum_{S \subseteq N: x \in S}\bmuone^{(i)}(S).
    \tag{$\bmuone^{(i)}([i+\lfloor\frac{n}{k}\rfloor])=1$ and $x \in [j] \subseteq [i + \lfloor \frac{n}{k}\rfloor]$}
\end{align*}
This inequality shows that agent $x$ will adopt for all $x\in[j]$. By the monotonicity of the sender's utility, the contribution to the sender's utility when prefix $[j]$ is realized under state $\omega_0$ is at least
\begin{align*}
    (1-\lambda)\cdot\bmuzero^{(i)}([j])\cdot V([j])=(1-\lambda)\cdot(\theta_j-\theta_{j+1})\cdot V([j]),
\end{align*}
which, after summing over $i\le j\le i+\lfloor\frac{n}{k}\rfloor$, completes the proof of the lemma.
\end{proof}

\subsection{Proof of \Cref{lemma:masking-matching-randomness}}
\label{app:proof-masking-match-random}
\maskMatchingRandomness*
\begin{proof}[Proof of \Cref{lemma:masking-matching-randomness}]
    We first prove that $a_i(s_i=1)=1$ for all $i\in N$ when no leakage is observed. We have
    \begin{align*}
        \sum_{s_{-i}}\bmuzero(1,s_{-i})
    =   \sum_{j=i}^{n}\bmuzero([j])
    =   c_0\sum_{j=i}^{n}(\theta_j - \theta_{j+1})
    =   c_0\cdot\theta_j
    \le c_0.
    \end{align*}
    On the other hand, for the probability under $\bmuone$, we have
    \begin{align*}
        \sum_{s_{-i}}\bmuone(1,s_{-i})\ge \bmuone(N)\ge1-c_1\ge c_0,
    \end{align*}
    where the last step applies the assumption that $c_0+c_1\le 1$. We thus have
    \[
    \sum_{s_{-i}}\bmuzero(1,s_{-i})\le\theta_i\cdot \sum_{s_{-i}}\bmuone(1,s_{-i}),
    \]
    which shows that agent $i$ will adopt.

    Now we verify the second property. For all $i\le m$ and $j\ge m$, we have
    \begin{align*}
        \bmuzero([j])=&\ c_0\cdot(\theta_j-\theta_{j+1})\\
        \le &\ c_1\cdot(\theta_j-\theta_{j+1}) \tag{$c_0\le c_1$}\\
        \le&\  \theta_i\cdot\frac{c_1}{\theta_m}\cdot(\theta_j-\theta_{j+1})
        \tag{$i\le m \implies \theta_i\ge\theta_m$}\\
        \le&\ \theta_i\cdot\bmuone([j]),
        \tag{definition of $\bmuone$}
    \end{align*}
    which proves the second property.
\end{proof}

\subsection{Proof of \Cref{thm:gap-expected-star}}
\label{app:proof-gap-expected-star}
\gapExpectedStar*
\begin{proof}[Proof of \Cref{thm:gap-expected-star}]
    We begin by illustrating the main idea of the proof in the special case of $k=n-1$, i.e., the center of the star observes all other agents' signal realization. We will then extend the proof to the general case for any $k$.

\paragraph{The case where $k=n-1$.}
We will show that the signaling scheme $(\bmuzero,\bmuone)$ described in \Cref{lemma:masking-matching-randomness} with cutoff $m=\lfloor \frac{n}{2}\rfloor+1$ and $c_0=c_1=\frac{1}{2}$ achieves a $4$-approximation to $\OPTprivate$.

Let $[j]$ be the realized prefix from $(\bmuzero,\bmuone)$, and let $i$ be the center of the $(n-1)$-star. We will prove that, for each $1\le j\le n$ and over the randomness of $i$, the probability that all agents in $[j]$ adopt is at least $\frac{1}{2}$. Consider the following two cases:
\begin{itemize}
    \item \textbf{Case 1:} $j\ge m$. Property~(2) from \Cref{lemma:masking-matching-randomness} guarantees that if $i\ge m$---which happens with probability $\frac{n-m+1}{n}\ge\frac{1}{2}$---the center $i$ will adopt. Moreover, property (1) of \Cref{lemma:masking-matching-randomness} guarantees that all agents in $[j]\setminus\{i\}$ will also adopt. Thus, with probability at least $\frac{1}{2}$, all agents in $[j]$ will adopt.
    \item \textbf{Case 2:} $j<m$. In this case, $i\ge m>j$ holds with probability at least $\frac{n-m+1}{n}\ge\frac{1}{2}$. When this happens, $i$ does not intersect with $[j]$, so property (1) again guarantees that all agents in $[j]$ will adopt.
\end{itemize}

Combining the above two cases, the sender's expected utility (over the randomness in both $i$ and $j$) is at least
\begin{align*}
    &~\frac{1}{2}\sum_{{j=1}}^{n}\left(\lambda\bmuone([j])+(1-\lambda)\bmuzero([j])\right)\cdot V([j])\\
\ge &~\frac{1}{2}\cdot\lambda\cdot\bmuone([n])\cdot V([n]) +\frac{1}{2}\cdot(1-\lambda)\cdot \sum_{j=1}^n \bmuzero([j])\cdot V([i])\\
\ge &~\frac{1}{2}\cdot\lambda\cdot(1-c_1)\cdot V([n]) +\frac{1}{2}\cdot(1-\lambda)\cdot \sum_{j=1}^n c_0(\theta_i-\theta_{i+1})\cdot V([i]) \tag{$\bmuone([n])\ge1-c_1$}\\
=   &~\frac{1}{4}\left(\lambda\cdot V([n]) + (1-\lambda)\cdot \sum_{j=1}^n(\theta_i-\theta_{i+1})\cdot V([i])\right) \tag{$c_0=c_1=\frac{1}{2}$}\\
=   &~\frac{1}{4}\cdot\OPTprivate. \tag{\cite[Theorem 1]{AB19}}
\end{align*}
We have thus established $\PoDR(\kstar)\le 4$ for $k = n - 1$.

\paragraph{The case with general $k$.}
Let $0<\alpha,\beta,\gamma<1$ with $\alpha+\beta+\gamma=1$ be parameters to be chosen later.
Consider the signaling scheme given by \Cref{lemma:masking-matching-randomness}, with parameters $m=\lfloor\alpha n\rfloor+1$ and $c_0=c_1=\frac{1}{2}$. Similar to the $k=n-1$ case, we aim to show that for each realized prefix $[j]$, over the randomness of the center $i$, there is a constant probability that all agents in the prefix $[j]$ choose to adopt.

Based on the parameters $\alpha,\beta,\gamma$, we divide the agents into three groups as follows: group $A$ includes agents indexed in $[1,\lfloor\alpha\cdot n\rfloor]$, group $B$ includes agents indexed in $\left[\lfloor\alpha\cdot n\rfloor+1, \lfloor\alpha\cdot n\rfloor+\lfloor\beta\cdot n\rfloor\right]$, and group $C$ includes agents indexed in $\ge\lfloor\alpha\cdot n\rfloor+\lfloor\beta\cdot n\rfloor+1$. The probability that a randomly sampled agent belongs to each of the three groups is approximately $\alpha,\beta$ and $\gamma$ (with an additive error of at most $\frac{1}{n}$).

Consider the following cases:
\begin{itemize}
    \item When $j$ comes from either group $A$ or group $B$: property (1) of \Cref{lemma:masking-matching-randomness} ensures that all agents in $[j]$ adopt if $i\not\in [j]$. This condition holds when $i$ belongs to group $C$, which happens with probability at least $\gamma-\frac{1}{n}$.
    \item When $j$ comes from group $C$: With probability at least $(\alpha - \frac{1}{n})\cdot(\beta - \frac{1}{n}) \ge \alpha\beta-\frac{1}{n}$, both of the following two events occur simultaneously: (i) $i$ belongs to group $A$; (ii) at least one of the $k$ leaves falls into group $B$ and leaks a signal of value $1$.

    When both events happen,
    the center $i$ has an updated belief that the realized prefix $[j]$ must lie within the interval $j\in[l,r]$, where $l$ is the largest index of an agent who leaks a signal of value $1$, and $r+1$ is the smallest index of an agent who leaks a signal of value $0$, or $r=n$ if no such agent exists. In particular, $j$ should belong to either group $B$ or $C$, thus satisfying $l\ge m$.

    According to property (2) of \Cref{lemma:masking-matching-randomness}, it holds for every $j' \in [l, r]$ that $\bmuzero([j'])\le\theta_i\cdot\bmuone([j'])$, which implies
    \[
        \sum_{j'\in[l,r]}\bmuzero([j'])\le\theta_i\cdot\sum_{j'\in[l,r]}\bmuone([j']).
    \]
    Therefore, receiver~$i$ will adopt. By property (1) of \Cref{lemma:masking-matching-randomness}, all agents in $[j]\setminus\{i\}$ also adopt. Therefore, with probability at least $\alpha\beta-\frac{1}{n}$, all agents in $[j]$ adopt.
\end{itemize}

The above case discussion implies that for any realized $[j]$, the probability that all agents in $[j]$ adopt is at least $\min\{\gamma,\alpha\beta\}-\frac{1}{n}$. If we set $\alpha=\beta=\sqrt{2}-1$ and $\gamma=3-2\sqrt{2}$, this probability reduces to $p^\star=3-2\sqrt{2}-\frac{1}{n}$. In particular, we have $p^\star\ge0.12$ when $n\ge 20$. When $n < 20$, applying \Cref{thm:upper-bound-private-public} gives a public signaling scheme that achieves an $O(n) = O(1)$ approximation to $\OPTprivate$. For $n\ge20$, following the a similar argument to the $k=n-1$ case, we can bound the sender's expected utility as
\begin{align*}
    &\ p^\star\sum_{j\ge m}\left(\lambda\bmuone([j])+(1-\lambda)\bmuzero([j])\right)\cdot V([j])\\
    \ge&\ p^\star\cdot\lambda\cdot(1-c_1)\cdot V([n])
    +p^\star\cdot(1-\lambda)\cdot \sum_{j=1}^n {{c_0}}(\theta_i-\theta_{i+1})\cdot V([i])
    \tag{$\bmuzero(N)\ge1-c_1$}\\
    =&\ \frac{p^\star}{2}\cdot\OPTprivate.
    \tag{\cite[Theorem 1]{AB19} and $c_0=c_1=\frac{1}{2}$}
\end{align*}
As a result, the price of robust downstream utility can be upper bounded by
\[
\PoDR(\kstar)=
\frac{\OPTprivate}{\OPTexpected(\kstar)}\le\frac{2}{p^\star}\le 17 =O(1).
\]
Therefore, for any values of $n$ and $k$, we conclude that $\PoDR(\kstar)=O(1)$.
\end{proof}

\subsection{Observations for Submodular Utilities}
\begin{proposition}
\label[proposition]{prop:expected-gap-star}
    For any $k$ and any instance where the sender's utility function is XOS, we have $\PoDR(\kstar)=O(1)$.
\end{proposition}
\begin{proof}
    We prove this result by directly applying the optimal private scheme $(\bmuzero^\star,\bmuone^\star)$. {All agents who are not the center of the star} will follow the signal, since no additional leakages are observed. Therefore, when subset $S$ is realized from $(\bmuzero^\star,\bmuone^\star)$ and agent~$i$ is the center of the star, the actual adopters after best response is either $S$ or $S\setminus\{i\}$.

    Since the sender's utility function is XOS, there exists $K$ linear functions $V^1,\ldots,V^K$ such that
    \begin{align*}
        V(S)=\max_{k\in[K]} V^k(S)=\max_{k\in[K]}\sum_{i\in S}v_i^k.
    \end{align*}
    For any $S\in 2^N$, we have
    \begin{align*}
        \Ex{ i\sim\unif(N)}{V(S\setminus \{i\})}
        =&\frac{1}{n}\sum_{i\in N}\max_{k\in [K]}\left(V^k(S)-v_i^k\cdot\1{i\in S}\right)\\
        \ge&\max_{k\in [K]}\frac{1}{n}\sum_{i\in N}\left(V^k(S)-v_i^k\cdot\1{i\in S}\right)\\
        =&\max_{k\in [K]}\left(1-\frac{1}{n}\right)V^k(S)
        \tag{$\sum_{i\in N}v_i^k\cdot\1{i\in S}=V^k(S)$}\\
        =&\left(1-\frac{1}{n}\right) V(S).
    \end{align*}
    Taking an expectation over the randomness of $S\sim (\bmuzero^\star,\bmuone^\star)$, the sender's expected downstream utility is at least
    \begin{align*}
        \left(1-\frac{1}{n}\right) \cdot\Ex{S\sim(\bmuzero^\star,\bmuone^\star)}{V(S)}
        \ge\Omega(1)\cdot\OPTprivate.
    \end{align*}
    We have thus proved that $\PoDR(\kstar)=O(1)$ for all XOS utility functions.
\end{proof}

\begin{proposition}
\label[proposition]{prop:expected-gap-clique}
    For any $k\le n-\Omega(n)$ and any instance where the sender's utility function is submodular or XOS, we have $\PoDR(\kclique)=O(1)$.
\end{proposition}
\begin{proof}
    The proof of this proposition is similar to that of \Cref{prop:expected-gap-star}, but in this case, all the $k$ agents within the clique may deviate. This results in a subsampling factor of $\gamma=1-\frac{k}{n}$ rather than $1-\frac{1}{n}$. This subsampling factor remains $\gamma=\Omega(1)$  when $k\le n-\Omega(n)$. Again, by applying \Cref{lemma:property-xos}, we can lower bound the sender's expected utility by $\gamma\cdot\OPTprivate$, which in turn upper bound the price of robust downstream utility by $\OPTprivate/\OPTexpected(\kclique)\le O(\frac{1}{\gamma})=O(1)$.
\end{proof}

\section{Details for Section~\ref{sec:technical-lower-bound}}
\label{app:lower-bound}

\subsection{Proof of \Cref{thm:lower-bound-persuasive-supermodular}}
\label{app:proof-lower-bound-persuasive-supermodular}
\lowerBoundPersuasiveSupermodular*
\begin{proof}[Proof of \Cref{thm:lower-bound-persuasive-supermodular}]
    We extend the proof of \Cref{thm:lower-bound-persuasive-supermodular} to cases where $2^k\le n$ and $2^k\ge n$.

    \paragraph{The case with $2^k<n$.} In this case, we start with the instance described in \Cref{ex:hard-instance-supermodular} with $n'=2^k$ agents, and append $n-n'$ dummy agents in addition. The dummy agents all have persuasion level $0$ and do not contribute to the sender's utility.

    Let $I'$ denote the instance with $n'$ agents and $I$ denote the expanded instance with $n$ agents.
    Since dummy the agents do not affect the sender's utility, they do not affect the sender's utility in the private setting, so we have $\OPTprivate(I)=\OPTprivate(I')=\Theta(n')=\Theta(2^k)$. However, the signals sent to the dummy agents may get leaked and observed by the non-dummy agents, potentially reducing the sender's utility in the $k$-worst-case persuasive setting. Therefore, we have $\OPTpersuasive_k(I)\le\OPTpersuasive_k(I')\le O(1)$. Combining these two observations, we obtain a lower bound on the price of robust persuasiveness for instance $I$: $\PoWR_k(I)\ge\Omega(2^k)$.

    \paragraph{The case with $2^k>n$.} Let $k'=\lfloor\log_2 n\rfloor$ and consider the instance $I'$ with $n'=2^{k'}=n^{\lfloor\log_2 n\rfloor}$ agents defined in \Cref{ex:hard-instance-supermodular}. As in the previous case, if $n'<n$, we construct instance $I$ with $n$ agents by introducing $n-n'$ dummy agents.

    For the sender's optimal private utility, we have $\OPTprivate(I)=\OPTprivate(I')=\Theta(n')=\Theta(n)$ as the dummy agents do not affect the sender's utility. In the $k$-worst-case persuasive setting, since $k'\le k$, any $k$-worst-case persuasive scheme is also $k'$-worst-case persuasive. We thus have
    $\OPTpersuasive_{k}(I)\le \OPTpersuasive_{k'}(I)\le \OPTpersuasive_{k'}(I')\le O(1)$, where the last step follows from the case when $n'=2^{k'}$. As a result, we have $\PoWR_k(I)=
    \OPTprivate(I)/\OPTpersuasive_k(I)\ge\Omega(n)$.

    Combining the above cases as well and the case of $n=2^k$ proved in \Cref{sec:lower-bound-persuasive}, we conclude that for general $(n,k)$, there always exists an instance such that
    \[
    \PoWR_k=\frac{\OPTprivate}{\OPTpersuasive_k}\ge\Omega\left(\min\left\{n,2^k\right\}\right).
    \]
\end{proof}

\subsection{Subcube Partition}
\begin{lemma}[Subcube Partition]
\label[lemma]{lemma:existence-of-subcube-partition}
    For every integer $k \ge 0$ and $n = 2^k$, there exists a subcube partition $\{C_0, C_1, \ldots, C_n\}$ of $\{0,1\}^n$ that satisfies:
    \begin{itemize}
        \item Each $C_i$ has a co-dimension of $\le k + 1$, i.e., $C_i$ is obtained by fixing at most $k+1$ different coordinates.
        \item $C_0$ is the subcube $\{x \in \{0, 1\}^n: x_1 = 0\}$.
        \item For each $i \in [n]$, $C_i$ is a subset of $\{x \in \{0, 1\}^n: x_i = 1, x_{i+1} = 0\}$ (where we regard $x_{n+1}$ as $0$ for all $x \in \{0, 1\}^n$).
    \end{itemize}
\end{lemma}
\begin{remark}[Equivalent interpretation of the subcube partition]
    The three conditions for the subcube partition can be equivalently interpreted as: for each $1\le i\le n$, there exists a subset of leaked signals $I_i$ with $(i+1,s_{i+1}=0)\in I_i$ and $|I_i|\le {k+1}$, such that
    \begin{align*}
        C_i=\left\{
            (s_1,\ldots,s_n)\in\{0,1\}^n\mid
            I_i\triangleright s_{-i}, s_i=1
        \right\}
    \end{align*}
\end{remark}

\begin{proof}[Proof of \Cref{lemma:existence-of-subcube-partition}]
    We describe a decision tree (equivalently, a query algorithm) whose leaves form the desired subcube partition. For each $(x_1,\ldots,x_n)\in\{0,1\}^n$, its leaf is determined as follows:
    \begin{itemize}
        \item First, we query $x_1$. If $x_1 = 0$, we end with a leaf labeled $C_0$.
        \item Otherwise, we have $x_1 = 1$ and $x_{n+1} = 0$ at this point. We use binary search to find $i \in [n]$ such that $x_i = 1$ and $x_{i+1} = 0$: We start with $(l, r) = (1, n + 1)$. As the first step, we query $x_{(l+r)/2} = x_{n/2+1}$. If $x_{n/2+1} = 0$, we continue with $(l, r) = (1, n/2 + 1)$; otherwise we repeat with $(l, r) = (n / 2 + 1, n + 1)$. Note that we always keep the invariant $(x_l, x_r) = (1, 0)$ and that $r - l$ is a power of $2$.
        \item After exactly $k = \log_2 n$ such queries, we end up with $r = l + 1$, at which point we are certain that $x_l = 1$ and $x_{l+1} = x_r = 0$, so we label the leaf with $C_l$.
    \end{itemize}
    In total, we make at most $k + 1$ queries, so the co-dimensions are upper bounded by $k + 1$.
\end{proof}

\subsection{Proof of \Cref{thm:lower-bound-clique}}
\label{app:proof-lower-bound-clique}
\hardInstanceKClique*
\lowerBoundClique*
\begin{proof}[Proof of \Cref{thm:lower-bound-clique}]
Similar to the proof of \Cref{thm:lower-bound-broadcast}, we will first upper bound the sender's expected utility conditioned on the realization of the clique (denoted by $\{v_1,\ldots,v_k\}$), and consider the randomness of the clique at the end.

Let $[j]$ be the prefix realized from $(\bmuzero, \bmuone)$.
According to the same argument as in the $k$-broadcast setting, each agent $i\in\{v_1,\ldots,v_k\}$ share the same belief that $j\in B_{l(j)}$ where $B_{l(j)}=[v_{l(j)},v_{l(j)+1})$ is the unique block that contains $j$, so they will adopt if and only if the following inequality holds:
\begin{align}
    \sum_{j'\in B_{l(j)}}\bmuzero([j'])\le
 \theta_i\cdot\sum_{j'\in B_{l(j)}} \bmuone([j']),
 \label{eq:cond-adopt-in-clique}
\end{align}

Let $i^\star$ be the largest index in $\{v_1,\ldots,v_k\}$ that satisfies \Cref{eq:cond-adopt-in-clique}, and let $\Tilde{i}$ be the largest index of the actual adopter in $N$ for which all its prefix $[\Tilde{i}]$ are adopters. We claim that
\[
\Tilde{i}\le i^\star+\max_{l}|B_l|.
\]
To see this, suppose that $i^\star = v_{l'}$ for some $l' \in [k]$. By the maximality of $i^\star$, receiver $v_{l' + 1}$ is not among the adopters. It follows that
\[
    \Tilde{i} \le v_{l' + 1} - 1 = v_{l'} + (v_{l' + 1} - v_{l'} - 1) \le i^\star + |B_{l'}| \le i^\star + \max_l|B_l|.
\]

Therefore, the sender's utility when prefix $[j]$ is realized is at most
\begin{align}
    V([\Tilde{i}])\le 2^{1+\lfloor\frac{\Tilde{i}}{B}\rfloor}
\le 2^{\lfloor\frac{i^\star}{B}\rfloor+{2}}\cdot 2^{\max_{l}\frac{|B_l|}{B}},
\label{eq:util-clique}
\end{align}
which holds for all $j$ in the same block.
In the second step above, we applied the inequality $\lfloor x + y\rfloor \le \lfloor x \rfloor + y + 1$.
Recall that $i^\star$ satisfies \Cref{eq:cond-adopt-in-clique}. Therefore, the contribution of prefixes in block $B_{l}$ to the sender's expected utility is upper bounded by
\begin{align*}
    \sum_{j'\in B_{l}}\bmuzero([j'])\cdot V([\Tilde{i}])&\le
    \theta_{i^\star}\sum_{j'\in B_{l}} \bmuone([j']) \cdot
    2^{\lfloor\frac{i^\star}{B}\rfloor+{2}}\cdot 2^{\max_{l'}\frac{|B_{l'}|}{B}}\\
    &= {4}\sum_{j'\in B_{l}}\bmuone([j'])\cdot 2^{\max_{l'}\frac{|B_{l'}|}{B}}\tag{
    $\theta_{i^\star}=2^{-\lfloor\frac{i^\star}{B}\rfloor}$
    }.
\end{align*}

Summing over all the blocks $B_0,\ldots,B_k$ and taking the expectation over its randomness, the sender's expected utility is at most
\begin{align}
   4\Ex{}{2^{\max_l\frac{|B_l|}{B}}}.
   \label{eq:expectation-block}
\end{align}

In the remainder of this proof, we will show that when the $k$-clique is uniformly drawn from $N$, the above expectation is upper bounded by $O(\log k)$. Combined with the fact that the sender's optimal private utility is $\OPTprivate=\Theta(k)$ in the instance \Cref{ex:hard-instance-expected}, this proves the $\Omega(k/\log k)$ lower bound.

Let us first consider a fixed $l\in[k]$ and assume that $v_l=x$ is fixed. Since $v_1,\ldots,v_k$ are sampled uniformly at random from all size-$k$ subsets of $N$, we have that the conditional distribution of $\{v_{l+1}, \ldots, v_k\}$ is uniform among size-$(k-l)$ subsets of $\{x+1, x+2, \ldots, n\}$. Therefore, for any length
$s\le n$, the probability that $|B_l|\ge s$ is at most the probability that none of the agents in the range $[x+1,x+s]$ are included in $\{v_{l+1}, \ldots, v_k\}$, which is upper bounded by
\begin{align*}
    \pr{}{|B_l|\ge s\mid v_l=x}
    \le&\frac{\binom{n-x-s}{k-l}}{\binom{n-x}{k-l}}.
\end{align*}
In addition, the probability that $v_l=x$ can be computed as
\begin{align*}
    \pr{\{v_{1},\ldots,v_k\}\sim\unif\left(\binom{N}{k}\right)}{v_l=x}=\frac{\binom{x-1}{l-1}\cdot\binom{n-x}{k-l}}{\binom{n}{k}}.
\end{align*}
As a result, the marginal probability that $|B_l|\ge s$ is upper bounded as
\begin{align*}
    \pr{}{|B_l|\ge s\mid v_l=x}=&\sum_{x}\pr{}{|B_l|\ge s\mid v_l=x}\cdot\pr{}{v_l=x}\\
    \le&\sum_x\frac{\binom{n-x-s}{k-l}}{\binom{n-x}{k-l}}\cdot
    \frac{\binom{x-1}{l-1}\cdot\binom{n-x}{k-l}}{\binom{n}{k}}\\
    =&\sum_x\frac{\binom{n-x-s}{k-l}\cdot\binom{x-1}{l-1}}{\binom{n}{k}}
    =\frac{\binom{n-s}{k}}{\binom{n}{k}}
    \tag{$\sum_{m=0}^n\binom{m}{j}\binom{n-m}{k-j}=\binom{n+1}{k+1}$}
    \\
    =& \frac{\left(n-s\right)^{\underline{k}}}{(n)^{\underline{k}}}
    \le \left(1-\frac{s}{n}\right)^k
    \le e^{-\frac{ks}{n}}.
\end{align*}
Since there are at most $\frac{n}{B}\le 5k$ blocks, from the union bound, we obtain
\begin{align}
    \pr{}{\max_l|B_l|\ge s}\le \min\{1,5k\cdot e^{-\frac{ks}{n}}\}
    \label{eq:tail-prob-bound}
\end{align}

We now turn to calculate the expectation in \Cref{eq:expectation-block}.
\begin{align*}
    \Ex{}{2^{\max_l\frac{|B_l|}{B}}}=&
    \sum_{s=1}^n 2^{\frac{s}{B}}\cdot\pr{}{\max_l|B_l|=s}\\
    =&\sum_{s=1}^n 2^{\frac{s}{B}}\left(\pr{}{\max_l|B_l|\ge s}-\pr{}{\max_l|B_l|\ge s+1}\right)\\
    \le& 2^{\frac{1}{B}}+\sum_{s\ge 2}\left(2^{\frac{s}{B}}-2^{\frac{s-1}{B}}\right)\pr{}{\max_l|B_l|\ge s}\\
    \le& 3+\sum_{s\ge 2}2^{\frac{s}{B}}\cdot\frac{\ln2}{B}\cdot \pr{}{\max_l|B_l|\ge s}\tag{
    $1-2^{-\frac{1}{B}}=1-e^{-\frac{\ln 2}B}\le\frac{\ln2}{B}$
    }\\
    \le&3+\frac{\ln2}{B}\sum_{s\ge 2}
    \min\{1,e^{\frac{s\ln2}{B}-\frac{ks}{n}+\log(5k)}\}
    \tag{From \cref{eq:tail-prob-bound}}\\
    \le&3+\frac{\ln2}{B}\sum_{s\ge 2}
    \min\{1,e^{-\frac{s}{B}+\log(5k)}\}
    \tag{$B=\frac{4n}{k}$}\\
    =&3+\frac{\ln2}{B}\left(\log(5k)\cdot B+\sum_{s\ge 1}
    e^{-\frac{s}{B}}\right)\\
    \le&3+\frac{2\ln 2}{B}\frac{1}{1-e^{-1/B}}+O(\log k)\le O(\log k).
    \tag{$1-e^{-1/B}\ge \frac{1}{2B}$ when $B\ge2$}
\end{align*}
As a result, we have proved that $\OPTexpected(\kclique)\le O(\log k)$ for all prefix-based schemes, which in turn establishes a lower bound of $\Omega(k/\log k)$ on the price of robust downstream utility.

\end{proof}

\subsection{Lower Bounds for Submodular Utilities}
\label{app:lower-bound-submodular}

In this section, we present an instance with an anonymous submodular utility function, and use it to establish an $\Omega(k)$ lower bound on the price of robust persuasiveness ($\PoWR_k$). A function $V$ is anonymous submodular if there exists a concave increasing function $f:\mathbb{N}\to\mathbb{R}$ such that $V(S)=f(|S|)$ for all $S$.\footnote{We focus on anonymous submodular utilities because their optimal private scheme is characterized in \cite{AB19}.}

We start with the extreme case of public persuasion where all signals are leaked. The following instance witnesses an $\Omega(n)$ gap between public and private persuasion.
\begin{example}
\label{ex:private-public-submodular}
All receivers $i\in N$ have the same persuasion level $\theta_i=\frac{1}{n}$. The sender's utility function is defined as $V_0(S)=f_0(|S|)=\1{|S|\ge1}$.
We also let $\lambda=\pr{}{\omega_1}=2^{-n}$ be very small so that the majority of the sender's utility is gained under state $\omega_0$.
\end{example}
\begin{theorem}
\label[theorem]{thm:gap-private-public-submodular}
    In \Cref{ex:private-public-submodular} where the sender's utility function is anonymous submodular, we have $\PoWR_{n-1}=\OPTprivate/\OPTpublic=\Omega(n)$.
\end{theorem}
\begin{proof}[Proof of \Cref{thm:gap-private-public-submodular}]
    Since the sender can achieve a utility of at most $\lambda \cdot V_0(N) = O(2^{-n})$ under state $\omega_1$, it suffices to consider the utility under state $\omega_0$ for both private and public persuasion. By~\cite[Theorem~2]{AB19}, the optimal private scheme $\bmuzero^\star$ only assigns non-zero probabilities to subsets of size $\theta=\sum_{i=1}^n\theta_i=1$ and achieves a utility of $\OPTprivate(\omega_0) = (1 - \lambda)\cdot f_0(\theta)=1 - \lambda$. We thus have $\OPTprivate(\omega_0)=\Theta(1)$.

    Now let $(\bmuzero,\bmuone)$ be the optimal public scheme. For each nonempty subset of adopters $S\neq\emptyset$, to make agent $i\in S$ adopt after observing the realization $S$, it should satisfy
    $\bmuzero(S)\le\theta_i\cdot\bmuone(S)=\frac{1}{n}\bmuone(S)$. This bounds $\OPTpublic(\omega_0)$ as
    \begin{align}
        \OPTpublic(\omega_0) \le \sum_{S \subseteq N: S \ne \emptyset}\bmuzero(S)\cdot V_0(S)\le \frac{1}{n}\sum_{S \subseteq N: S \ne \emptyset} \bmuone(S) \le \frac{1}{n}.
      \label{eq:public-submodular}
    \end{align}
    Thus, we have $\PoWR_{n-1}=\OPTprivate/\OPTpublic = \Omega(n)$ on this instance.
\end{proof}

In the following, we move on to $k$-worst-case persuasive schemes for general $k$. We show how to establish an $\Omega(k)$ price of robust persuasiveness by embedding the hard instance in \Cref{ex:private-public-submodular} with $k$ agents into a general instance with $n$ agents.

\begin{example}[Hard instance for anonymous submodular utilities]
    \label{ex:persuasive-submodular}
    All receivers $i\in N$ have the same persuasion level $\theta_i=\frac{1}{k}$. The sender's utility function is defined as
    \[
    V(S)=\Ex{S'\sim\binom{N}{k}}{V_0(S'\cap S)}=\Ex{S'\sim\binom{N}{k}}{f_0(|S'\cap S|)},
    \]
    where $S'$ is uniformly drawn from all subsets of size $k$ and $V_0,f_0$ are defined in \Cref{ex:private-public-submodular}.
    We also set $\lambda=\pr{}{\omega_1}=2^{-n}$.
\end{example}

\begin{restatable}{theorem}{thmSubmodularPersuasiveLowerBound}
\label[theorem]{thm:lower-bound-persuasive-submodular}
    For any $k$, the instance in \Cref{ex:persuasive-submodular} has $\PoWR_k = \Omega(k)$.
\end{restatable}

\begin{proof}[Proof of \Cref{thm:lower-bound-persuasive-submodular}]
    We first verify that $V(S)$ is indeed anonymous submodular. We have
    \begin{align*}
    V(S)=\pr{S'\sim\binom{N}{k}}{S'\cap S\neq \emptyset}=1-\frac{\binom{n-|S|}{k}}{\binom{n}{k}}=:f(|S|).
\end{align*}
It is not hard to see that $f(\cdot)$ is concave and increasing. Therefore, according to \cite[Theorem~2]{AB19}, the optimal private scheme $\bmuzero^\star$ assigns non-zero probabilities only to subsets $S$ such that $|S|=\lceil\frac{n}{k}\rceil$ or $|S|=\lfloor \frac{n}{k}\rfloor$. In either case, $f(|S|)$ is at least
\begin{align*}
    f(\lfloor n/k \rfloor)
=   1-\frac{\binom{n-\lfloor n/k\rfloor}{k}}{\binom{n}{k}}
\ge 1 - \left(\frac{n - \lfloor n/k \rfloor}{n}\right)^k
\ge 1 - \exp\left(-\frac{\lfloor n/k\rfloor}{n}\cdot k\right)
\ge 1 - e^{-1/2}.
\end{align*}
Thus, we have $\OPTprivate=\Theta(1)$.

It remains to show that no $k$-worst-case persuasive signaling scheme can achieve a utility better than $1/k$. Again, by our choice of $\lambda = 2^{-n}$ in \Cref{ex:persuasive-submodular}, the utility gained under state $\omega_1$ is at best $\lambda\cdot V_0(N) = 2^{-n} \ll 1/k$, so we focus on the utility gained under state $\omega_0$ in the following. For the sake of contradiction, assume that $(\bmuzero,\bmuone)$ is $k$-worst-case persuasive and satisfies
\begin{align*}
    \Ex{S\sim \bmuzero}{V(S)}=\Ex{S'\sim\binom{N}{k}}{\Ex{S\sim\bmuzero}{V_0(S\cap S')}}>\frac{1}{k}.
\end{align*}
Then, there must be a subset $S'\subseteq N$ of size $k$ that satisfies $\Ex{S\sim\bmuzero}{V_0(S\cap S')}\ge\frac{1}{k}$. Consider the instance restricted to $S'$ (which is equivalent to the instance in \Cref{ex:private-public-submodular} of size $k$), and define the signaling scheme on $S'$ using the marginal distributions of $\bmuzero'=\bmuzero|_S,\bmuone'=\bmuone|_S$. Since $(\bmuzero,\bmuone)$ is $k$-worst-case persuasive, it follows that $(\bmuzero',\bmuone')$ is publicly persuasive on $S'$. Moreover, the sender's utility under $(\bmuzero',\bmuone')$ is
\begin{align*}
    \Ex{S\sim\bmuzero'}{V_0(S)}=\Ex{S\sim\bmuzero}{V_0(S\cap S')}>\frac{1}{k},
\end{align*}
which contradicts the upper bound derived in \Cref{eq:public-submodular} for the public setting. Therefore, any $k$-worst-case persuasive scheme on $N$ must have a utility of at most $(1 - \lambda)\cdot\Ex{S \sim \bmuzero}{V(S)} \le 1/k$ under state $\omega_0$. This proves $\OPTpersuasive_k = O(1/k)$ and establishes the lower bound $\PoWR_k = \Omega(k)$ for anonymous submodular functions.
\end{proof}

\section{Direct vs Indirect Schemes}\label{sec:direct-vs-indirect}
In this appendix, we give a small instance---with $n = 3$ receivers and a degenerate distribution $\cG$ (i.e., the leakage pattern is deterministic)---on which: (1) a signaling scheme with a size-$3$ signal space is able to match $\OPTprivate$, the utility of the optimal scheme in the private setting; (2) no scheme with a size-$2$ signal space can match $\OPTprivate$. The first claim will be shown by constructing the scheme directly, while the second claim was verified by an exhaustive search on a computer.

\subsection{A Somewhat-Indirect Scheme}\label{sec:somewhat-indirect}
To gain intuition, we start with an instance on which a size-$2$ signal space \emph{is} sufficient to achieve the optimal utility in the expected utility robustness model, but the resulting scheme is necessarily \emph{indirect}, in the sense that the recommended action for an receiver depends on not only their own signal, but all the observable signals after the leakage.

\paragraph{The instance.} There are $n = 3$ receivers. The utility of the sender is given by $V(S) = |S|$. The prior probability of state $\omega_1$ is $\lambda = 1/2$, and the persuasion levels are $(\theta_1, \theta_2, \theta_3) = (0, 0, 1/2)$. In other words, receivers $1$~and~$2$ are very hard to persuade in the sense that they play the positive action only if they are absolutely certain that the state is $\omega_1$. Finally, the distribution $\cG$ over leakage patterns is defined such that receivers $1$~and~$2$ see the signals of each other, while receiver~$3$ sees the signal of either receiver~$1$ or receiver~$2$, each with probability $1/2$.

\paragraph{When no signal is leaked to receiver~$3$.} Without the leakage to receiver~$3$, there is a simple direct scheme that is optimal: Under state $\omega_1$, we recommend action~$1$ to everyone. Under $\omega_0$, we always recommend $0$ to receivers $1$~and~$2$, and we recommend $1$ to receiver $3$ with probability exactly $1/2$. Formally:
\[
    \mu_1(1, 1, 1) = 1, \mu_0(0, 0, 1) = \mu_0(0, 0, 0) = 1/2.
\]
It is easy to verify that the scheme above is persuasive, and the resulting utility is
\[
    \lambda\cdot 3 + (1 - \lambda)\cdot\frac{1}{2}\cdot 1
=   \frac{1}{2}\cdot 3 + \frac{1}{2}\cdot\frac{1}{2}\cdot 1 = \frac{7}{4}.
\]
However, if receiver~$3$ can see the signal sent to either receiver~$1$ or receiver~$2$, the scheme above becomes no longer persuasive, as the leaked signal would reveal the state.

\paragraph{An optimal scheme.} Here is a different scheme that achieves the same utility of $7/4$, even when there is an additional leakage to receiver~$3$:\footnote{We rename the signal space as $\{+, -\}$ (rather than $\{0, 1\}$), as the signals no longer correspond to the actions.}
\[
    \mu_1(+,+,+) = \mu_1(-,-,+) = 1/2,
\]
\[
    \mu_0(+,-,+) = \mu_0(-,+,+) = \mu_0(+,-,-) = \mu_0(-,+,-) = 1/4.
\]

From the perspectives of receivers $1$~and~$2$, whether their signals match reveals the state: they always match when the state is $\omega_1$, and they differ under state $\omega_0$. Therefore, both of them would play action $1$ if and only if the state is $\omega_1$.

From the perspective of receiver~$3$, when they receive signal ``$+$'', observing an additional signal (leaked by either receiver~$1$ or receiver~$2$) provides no additional information about the state. It follows that receiver~$3$ will follow the recommendation when their signal is ``$+$''. It can then be verified that the optimal utility of $7/4$ is achieved by the scheme above.

\paragraph{Remarks.} Note that in this scheme, the signal sent to either receiver~$1$ or receiver~$2$, when viewed alone, does not determine the recommended action. Both signals must be observed in order to obtain information about the state and thus the recommended action. This indirection is crucial for hiding the true state from receiver~$3$, who always observes one of the two signals.

Also, this gives a simple instance on which leakages might \emph{benefit} the sender---if we remove the edges between receivers $1$~and~$2$, the signaling scheme no longer works, and the resulting utility of the sender would be strictly lower than $\OPTprivate$. This non-monotonicity has been observed in the prior work of~\cite{KT23}, witnessed by a slightly larger instance.

\subsection{The Actual Instance and Three-Signal Scheme}
Now we introduce the actual instance that separates signaling schemes with size-$2$ and size-$3$ signal spaces when the expected downstream utility is concerned.

\paragraph{The instance.} Again, we have $\lambda = 1/2$, and there are $n = 3$ receivers with persuasion levels $(\theta_1, \theta_2, \theta_3) = (3/4, 2/4, 1/4)$. The sender's utility function is $V(S) = \1{1 \in S} + \1{1, 2 \in S} + \1{S = [3]}$. In other words, $V(S)$ is the length of the longest prefix that is contained in $S$. The three receivers form a cycle and each receiver observes the signal sent to the next receiver in the cycle. More formally, the leakage graph consists of the edges $(2 \to 1)$, $(3 \to 2)$ and $(1 \to 3)$.

\paragraph{The optimal private scheme.} By the results of~\cite{AB19}, in the private Bayesian persuasion setup, the optimal signaling scheme is given by
\[
    \mu_1(1, 1, 1) = 1,
\]
\[
    \mu_0(0, 0, 0) = \mu_0(1, 0, 0) = \mu_0(1, 1, 0) = \mu_0(1, 1, 1) = 1/4.
\]
The resulting utility is
\[
    \lambda\cdot 3 + (1 - \lambda)\cdot\frac{1}{4}\cdot(0 + 1 + 2 + 3) = \frac{9}{4}.
\]

Note that in the scheme above, we ensure that each receiver~$i$, conditioning on the state being $\omega_0$, plays action~$1$ with probability exactly $\theta_i$, which is the highest possible. Furthermore, the three receivers are coordinated perfectly in the sense that the resulting action profile always forms a prefix.

\paragraph{A three-signal scheme.} Next, we show how we can match this optimal utility of $9/4$ using a size-$3$ signal space when leakages are present. For clarity, we use $\{+, -, 0\}$ as the signal space for receivers $1$~and~$3$, and $\{X, Y\}$ for receiver $2$. The signaling scheme is defined as:
\[
    \mu_1(+, X, +) = \mu_1(-, X, -) = 1/2;
\]
\[
    \mu_0(+, X, +) = \mu_0(-, X, -) = 1/8,
\]
\[
    \mu_0(+, X, -) = \mu_0(-, X, +) = 1/8,
\]
\[
    \mu_0(+, X, 0) = \mu_0(-, X, 0) = 1/8,
\]
\[
    \mu_0(+, Y, 0) = 1/4.
\]
The following are the intended best responses of the receivers:
\begin{itemize}
    \item Receiver~$1$ plays action~$1$ if and only if the signal for receiver~$2$ is $X$.
    \item Receiver~$2$ plays action~$1$ if and only if their own signal is $X$, and the signal of receiver~$3$ is not $0$ (i.e., either $+$ or $-$).
    \item Receiver~$3$ plays action~$1$ if and only if their own signal is the same as the one for receiver~$1$.
\end{itemize}
Note that for receivers $1$~and~$3$, we use the same ``matching signs'' trick as in Section~\ref{sec:somewhat-indirect}.

Intuitively, in the scheme above, it is crucial to send three different signals to receiver~$3$---the signals $+$ and $-$ allow us to prescribe the action for receiver~$3$ without leaking too much information to receiver~$2$ (who observes the signal sent to receiver~$3$). Then, the third signal ``$0$'' allows us to recommend action~$0$ to receiver~$2$, without leaking information to receiver~$1$.

\paragraph{Two-signal schemes.} A brute-force search over signaling schemes with a size-$2$ signal space suggests that the expected utility is at best $17 / 8 < 9 / 4$. This, for example, is obtained by the following scheme:
\[
    \mu_1(+, X, +) = \mu_1(-, X, -) = 1/2;
\]
\[
    \mu_0(+, X, +) = \mu_0(-, X, -) = 1/8,
\]
\[
    \mu_0(+, X, -) = \mu_0(-, X, +) = 1/8,
\]
\[
    \mu_0(+, Y, -) = \mu_0(-, Y, +) = 1/4.
\]

In more detail, with size-$2$ signal spaces, there are only $2^n = 8$ different signal profiles. When each signal profile is sent to the receivers, there are $2^n = 8$ possible combinations of the best responses. As discussed in Section~\ref{sec:discussion}, we perform a brute-force search over the $8^8 = 2^{24}$ different possibilities of the best responses under each signal profile. For each possibility, maximizing the expected utility becomes a linear program over the space of $(\bmuzero, \bmuone)$, which has a dimension of only $2\times 2^n = 16$. Therefore, we can compute the optimal two-signal scheme in a reasonable amount of time.

\section{The Supermodular Case with Externalities}
\label{sec:externality}

In the previous sections, we focused on a setting without externalities among receivers: we assumed that each receiver $i$ adopts if and only if the posterior probability of state {$\omega_0$} is no less than their threshold $p_i$, independently of the actions taken by other receivers. In this section, we show that when there are externalities among receivers, the gap between $\OPTprivate$ and $\OPTpersuasive_1$ can be unbounded even on instances that are tractable in the fully private setting~\cite{AB19}.

\begin{example}[Hard instance with externality among receivers]
    \label[example]{ex:externality}
    Let $n\ge2$ be the number of agents and $\eps>0$ be any small constant.
    The prior distribution is defined as $\lambda=\pr{}{\omega_1}=\eps$.
    Suppose agent $1$ is a ``special agent'' whose action dictates the sender's utility:
    \begin{align*}
        V(S)=\1{1\in S}.
    \end{align*}
    The persuasion levels are defined as $\theta_1=1$ and $\theta_i=\eps$ for all $i\ge 2$. Under a signaling scheme $\bmu=(\bmuzero,\bmuone)$, the action of each receiver $i\in N$ upon observing $s_i$ and leakages $I_i$ is defined as follows:
    \begin{align*}
        a_i^{\bmu}(s_i,I_i)\coloneqq
        \1{
\sum_{s_{-i}:I_i\triangleright s_{-i}}\bmuzero(s_i,s_{-i})\le\theta_i\cdot{\bmuone(s_i,\boldsymbol{1})}\cdot\1{I_{i}\triangleright\boldsymbol{1}}}.
    \end{align*}
    This action can be viewed the best response under the following utility function $u_i:\Omega\times\A^n\to\mathbb{R}$:
    \begin{align*}
        &\forall \omega\in\Omega,\ s_{-i}\in\{0,1\}^{n-1},\quad
        u_i(\omega,s_i=0,s_{-i})=0;\\
        &\forall s_{-i}\in\{0,1\}^{n-1},\quad
        u_i(\omega_0,s_i=1,s_{-i})=-\frac{\lambda}{1-\lambda}\frac{1}{\theta_i};\\
        &\forall s_{-i}\in\{0,1\}^{n-1},\quad
        u_i(\omega_1,s_i=1,s_{-i})=\1{s_{-i}=\boldsymbol{1}}.
    \end{align*}
    In other words, receiver $i$'s utility is always $0$ if they do not adopt. When they choose to adopt, the utility has no externality under state $\omega_0$, and under state $\omega_1$, the utility is only nonzero when all other agents also choose to adopt.
    \end{example}

\begin{theorem}
\label[theorem]{thm:externality}
    For any $\eps>0$ and in the $1$-worst-case persuasiveness setting, the instance described in \Cref{ex:externality} satisfies $\PoWR_1\ge \Omega(\frac{1}{\eps})$.
    Letting $\eps\to0$ shows that the price of worst-case robustness can be unbounded even when $k=1$.
\end{theorem}
\begin{proof}[Proof of \Cref{thm:externality}]
    We first show that the optimal private utility is $\OPTprivate=\Theta(1)$. Since both the sender and the receiver's utility functions are supermodular, according to \cite[Theorem 5]{AB19}, the optimal private scheme $(\bmuzero^\star,\bmuone^\star)$ is given by
    \begin{align*}
        \bmuone^\star(N)=1,
        \quad
        \bmuzero^\star(\{1\})=1-\eps,
        \ \bmuzero^\star(N)=\eps.
    \end{align*}
    This achieves an optimal utility of $\OPTprivate=1$ because agent $1$ can always be persuaded to adopt under both $\bmuzero^\star$ and $\bmuone^\star$.

    Now we turn to analyze $\OPTpersuasive_1$. Let $(\bmuzero,\bmuone)$ be the optimal $1$-worst-case persuasive scheme that achieves $\OPTpersuasive_1$. We claim that $\support(\bmuzero)\subseteq \{N,\emptyset\}$. Suppose towards a contradiction that there exists a non-empty, proper subset $S\subsetneq N$ such that $\bmuzero(S)>0$. Then, consider an agent $i\in S$ who receives signal $s_i=1$ and observes a leaked signal $I_i=\{(j,s_j=0)\}$ for $j\in N\setminus S$. Since $I_i$ is not consistent with the all-one signal, we have
    \begin{align*}
        \sum_{s_{-i}:I_i\triangleright s_{-i}}\bmuzero(s_i,s_{-i})
        \ge \bmuzero(S)>0=\theta_i
        \cdot{\bmuone(s_i,\boldsymbol{1})}\cdot\1{I_{i}\triangleright\boldsymbol{1}}.
    \end{align*}
    Therefore, agent $i$ will not adopt, thus violating the worst-case-persuasiveness of $(\bmuzero,\bmuone)$.

    As a result, we must have $\support(\bmuzero)\subseteq \{N,\emptyset\}$. Since all but the special agent has a very small persuasion level of $\theta_i=\eps$, to ensure that they follow the signal when the realized signal is $N$, we must have
    \begin{align*}
        \bmuzero(N)\le \eps\cdot\bmuone(N)\le\eps.
    \end{align*}
    As a result, the sender's expected utility is
    \begin{align*}
        \OPTpersuasive_1
        \le\lambda+(1-\lambda)\bmuzero(N)\le 2\eps.
    \end{align*}
    Therefore, we have $\OPTpersuasive_1\le O(\eps)$. This establishes that, when externalities are present, the price of worst-case robustness is lower bounded by
    \[
    \PoWR_1=\frac{\OPTprivate}{\OPTpersuasive_1}\ge {\Omega}\left(\frac{1}{\eps}\right).
    \]
\end{proof}

\end{document}